\documentclass[12pt]{article}
\usepackage{amsmath,amssymb,graphicx,alltt,xy,amsthm}
\date{}
\input xy
\xyoption{all}
\newtheorem{thrm}{Theorem}
\newtheorem{lem}{Lemma}

\newtheorem{cor}[thrm]{Corollary}

\newtheorem{Def}[thrm]{Definition}

\newtheorem{conj}{Conjecture}

\newcommand{\Zc} {Z_{\rm c}}

\newcommand{\Remark}{\vspace{0mm} \parindent=0pt
         {\bf Remark.} \hspace{0mm} \parindent=3ex}

\newcommand{\Beginproof}{\vspace{0mm} \parindent=0pt
         {\bf Proof.} \hspace{3mm} \parindent=3ex}
\newcommand{\Endproof}{$\Box$ \vspace{5mm}
                        \parindent=3ex}

\newcommand{\Pcr}{P_{\rm cr}}

\begin{document}
\title{Positive and necklace solitary waves on bounded domains}
\author{G. Fibich\thanks{%
School of Mathematical Sciences, Tel Aviv University, Tel Aviv 69978 Israel,
fibich@tau.ac.il (corresponding author)} \and D. Shpigelman\thanks{%
School of Mathematical Sciences, Tel Aviv University, Tel Aviv 69978 Israel,
dimashpi87@gmail.com}}
\maketitle

\begin{abstract}
We present new solitary wave solutions of the two-dimensional nonlinear Schr\"odinger equation on bounded domains
(such as rectangles, circles, and annuli). These multi-peak ``necklace'' solitary waves
consist of several identical positive profiles (``pearls''), such that adjacent ``pearls'' have opposite signs.
They are stable at low powers, but become unstable at powers well below the critical power for collapse~$\Pcr$.
This is in contrast with the ground-state (``single-pearl'') solitary waves on bounded domains,
which are stable at any power below~$\Pcr$.

On annular domains, the ground state solitary waves are radial at low powers,
but undergo a symmetry breaking at a threshold power well below~$\Pcr$.
As in the case of convex bounded domains,
necklace solitary waves on the annulus are stable at low powers and
become unstable at powers well below~$\Pcr$. Unlike on convex bounded domains, however,
necklace solitary waves on the annulus have a second stability regime at powers well above~$\Pcr$.
For example, when the ratio of the inner to outer radii is 1:2,
four-pearl necklaces are stable when their power is between $3.1\Pcr$ and $3.7\Pcr$.
This finding opens the possibility to propagate localized laser beams with substantially more power than was possible until now.

 The instability of necklace solitary waves is excited by perturbations that break the antisymmetry between adjacent pearls, and is manifested by power transfer between pearls. In particular, necklace instability is unrelated to collapse.
In order to compute numerically
the profile of necklace solitary waves
on bounded domains, 
we introduce a non-spectral variant of Petviashvili's renormalization method.
\end{abstract}

{\bf Keywords;} solitary waves; nonlinear schrodinger equation; bounded domains; nonlinear optics; stability 

\tableofcontents

\section{Introduction}
\label{chap:Introduction}

 The nonlinear Schr\"odinger equation (NLS) in free space
\begin{subequations}
  \label{eq:NLS_free_space}
\begin{equation}
i\psi_z(z,x,y)+\Delta\psi+|\psi|^{2}\psi=0, \qquad -\infty<x,y<\infty, \quad z>0,
\end{equation}
\begin{equation}
\psi(0,x,y) = \psi_0(x,y), \qquad -\infty<x,y<\infty
\end{equation}
\end{subequations}
is one of the canonical nonlinear equations in physics. In nonlinear
optics it models the propagation of intense laser beams in a bulk Kerr
medium. In this case, $z$~is the axial coordinate in the direction
of propagation, $x$ and~$y$ are the spatial coordinates in the
transverse plane, $\Delta \psi:=\frac{\partial^2}{\partial x^2}\psi+\frac{\partial^2}{\partial y^2}\psi$ is the
diffraction term, and~$|\psi|^2\psi$ describes the nonlinear Kerr
response of the medium. For more information on the NLS in nonlinear optics and on NLS theory in free space and
on bounded domains, see the recent book~\cite{NLS-book}.

In some applications, it is
desirable to propagate laser beams over long distances.
In theory, this can be done by the NLS solitary waves
$\psi_{\rm sw} = e^{i \mu z} R_\mu(x,y)$, where $R_\mu$ is a solution of
\begin{equation}
   \label{eq:R-2D-free}
\Delta R(x,y)-\mu R+|R|^{2}R=0,\qquad -\infty<x,y<\infty.
\end{equation}
Unfortunately, the solitary wave solutions of~\eqref{eq:NLS_free_space} are unstable,
so that when perturbed, they either scatter (diffract) as $z \to \infty$, or collapse at a finite distance~$\Zc<\infty$.

In order to mitigate this ``dual instability'' limitation, Soljacic, Sears, and Segev~\cite{Soljacic-98} proposed in~1998
to use a necklace configuration that consists of $n$~identical beams (``pearls'')
that are located along a circle at equal distances, such that
adjacent beams are out of phase (i.e., have opposite signs),
see Figure~\ref{fig:segev_necklace_initial condition_8_pearls_no_green_lines}.
 The idea behind this setup is that the repulsion between adjacent out-of-phase beams resists the diffraction of each beam, and thus slows down its expansion.
Necklace beams in a Kerr medium were first observed experimentally by Grow et al.~\cite{Gaeta-07}.
Necklace beams were also studied in~\cite{Desyatnikov-01,PhysRevA.80.053816,PhysRevLett.86.420,PhysRevLett.94.113902,Yang-05}.
Recently, Jhajj et al.\ used a necklace-beam configuration to set up a thermal waveguide in air~\cite{Jhajj-14}.

\begin{figure}[ht!]
\begin{center}
\scalebox{.6}{\includegraphics{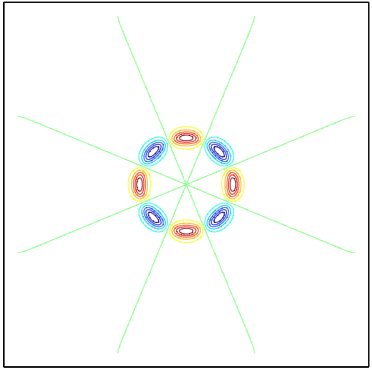}}
\caption{A circular necklace beam with 8~pearls (beams). Adjacent pearls are identical but have opposite phases. The electric field vanishes on the rays (solid lines) between adjacent pearls.
}
\label{fig:segev_necklace_initial condition_8_pearls_no_green_lines}
\end{center}
\end{figure}

As we shall see, there are no necklace solitary wave solutions of the free-space NLS~\eqref{eq:NLS_free_space}.
Thus, in a bulk  medium all
necklace beams ultimately collapse or scatter.
Yang et al.~\cite{Yang-05} showed theoretically and experimentally
that solitary necklace solutions can exist in a bulk medium with an optically-induced photonic lattice.
Because of the need to induce a photonic lattice, however, this approach is not applicable to propagation is a Kerr medium.
In this study we show that necklace solitary waves exist in a Kerr medium,
provided the beam is confined to a bounded domain.
This setup corresponds to propagation in hollow-core fibers, and is therefore relatively
easy to implement experimentally.

In hollow-core fibers, beam propagation can be modeled by the NLS on a bounded domain
\begin{subequations}
\label{eq:NLS_bounded}
\begin{equation}
i\psi_z(z,x,y)+\Delta\psi+|\psi|^{2}\psi=0,
\qquad (x,y)\in D , \qquad z>0,
\end{equation}
subject to
an initial condition
\begin{equation}
\psi_0(0,x,y)=\psi_0(x,y),
\qquad (x,y)\in D,
\end{equation}
and
a Dirichlet boundary condition at the fiber wall
\begin{equation}
\psi(z,x,y)=0, \qquad
(x,y)\in\partial D, \quad z\geq0.
\end{equation}
\end{subequations}
Here $D \subset {\Bbb R}^2$ is the cross section of the fiber, which is typically a circle of radius~$\rho$,
 denoted henceforth by~$B_\rho$.

Eq.~(\ref{eq:NLS_bounded}) admits the solitary waves $\psi_{\rm sw} =e^{i\mu z}R_\mu(x,y)$, where~$R_\mu$ is a solution of
\begin{subequations}
    \label{eq:dDim_R_ODE2}
\begin{equation}
 \label{eq:dDim_R_ODE2a}
\Delta R(x,y)-\mu R+|R|^{2}R=0, \qquad (x,y) \in D,
\end{equation}
\begin{equation}\label{eq:dDim_R_ODE2b}
R(x,y) =0, \qquad (x,y)  \in\partial D.
\end{equation}
\end{subequations}
In free space, the solitary-wave profile [i.e., the solution of~\eqref{eq:R-2D-free}] represents a perfect balance between the focusing nonlinearity and diffraction.
On a bounded domain, the reflecting boundary ``works with'' the focusing nonlinearity and ``against'' the diffraction.
In fact, the reflecting boundary can support finite-power solitary waves
even in the absence of a focusing linearity. These linear modes are solutions of the eigenvalue problem
\begin{equation}
  \label{eq:lin_R_mu_in_D}
\Delta Q(x,y) = \mu Q, \quad (x,y) \in D,
  \qquad\qquad
Q(x,y) =0, \quad (x,y)  \in\partial D.
\end{equation}
Of most importance is the first eigenvalue of~\eqref{eq:lin_R_mu_in_D} and its corresponding
positive eigenfunction, which we shall denote by~$\mu_{\rm lin}$ and~$Q^{(1)}$, respectively.

Solitary wave solutions of~\eqref{eq:NLS_bounded}  were studied
by Fibich and Merle~\cite{bounded-01}, Fukuizumi, Hadj Selem, and Kikuchi~\cite{Fukuizumi-12},
and Noris, Tavares, and Verzini~\cite{Noris-15},
primarily when~$D$ is the unit circle~$B_1$ and~$R_\mu$ is radial,
i.e., $R_\mu = R_\mu(r)$, $r = \sqrt{x^2+y^2}$, and $0 \le r\le 1$.
 In that case, for any $\mu_{\rm lin}<\mu<\infty$,
there exists a unique positive solution~$R_\mu^{(1)}(r)$.
This solution is monotonically decreasing in~$r$,
and its power $P(R_\mu^{(1)}):= \int_{B_1} |R^{(1)}_{\mu}|^2 \, dx dy$ is monotonically increasing in~$\mu$ from $P=0$ at
$\mu  =   \mu_{\rm lin}+$  to $P = \Pcr$ as $\mu \to \infty$.
In addition to this ground state, there exist a countable number of
excited radial states~$\{R_\mu^{(n)}(r)\}_{n = 2}^\infty$,
which are non-monotone and change their sign inside~$B_1$. These
excited states have a unique global maximum at $r=0$, and
additional lower peaks on concentric circles inside~$B_1$.

The excited states~$\{R_\mu^{(n)}(r)\}_{n = 2}^\infty$
are the two-dimensional radial analog of the excited states of the one-dimensional NLS on an interval,
see equation~\eqref{eq:NLS_bounded-1D} below,
which were studied by Fukuizumi et al.~\cite{Fukuizumi-12}.
 In our study here we consider a different type of solitary waves
of~(\ref{eq:NLS_bounded}), which attain their global maximum at $n$~distinct points inside~$D$.
These necklace solitary waves are thus the  non-radial
two-dimensional analog of the one-dimensional excited states.

The paper is organized as follows. In Section~\ref{sec:free-space} we briefly consider necklace solutions
in~${\Bbb R}^2$, which correspond to propagation in a bulk medium.
We illustrate numerically that their expansion is slower than that of single-beam solutions, and  show that 
there are no necklace solitary waves in free space.
In Section~\ref{sec:review} we briefly review the theory for
 the NLS on bounded domains.
 In Section~\ref{sec:Solitary wave} we construct necklace solitary waves with $n$~pearls (peaks), denoted by~$R_{\mu}^{(n)}$, on rectangular, circular, and annular domains.
To do that, we first compute the single-pearl (single-peak/ground state) solitary wave of~\eqref{eq:dDim_R_ODE2}, denoted by~$R_{\mu}^{(1)}$,
on a square, a sector of a circle, and a sector of an annulus, respectively.
Our numerical results for 
single-pearl solutions of~\eqref{eq:dDim_R_ODE2} suggest that:\footnote{These results are consistent with those obtained for radial positive solitary waves on the circle~\cite{bounded-01}
and for positive one-dimensional solitary waves on an interval~\cite{Fukuizumi-12}.}
\begin{enumerate}
  \item $R_{\mu}^{(1)}$ exists for $\mu$ in the range  $\mu_{\rm lin}<\mu< \infty$.
  \item As $\mu \to \mu_{\rm lin}$,  $R_{\mu}^{(1)}$ approaches the positive linear mode~$Q^{(1)}$, i.e., $R_{\mu}^{(1)} \sim c(\mu) Q^{(1)} $, where $ c(\mu) \to 0$.
  \item As $\mu$~increases, $R^{(1)}_{\mu}$ becomes
more localized, the effect of the nonlinearity becomes more pronounced, and that of the reflecting boundary becomes less pronounced.
  \item In particular, as $\mu \to \infty$,  $R^{(1)}_{\mu}$ approaches the free-space ground state~$R_{\mu, \rm 2D}^{(1), \rm free}$, which is the positive solution of~\eqref{eq:R-2D-free}.
  \item  The pearl power $P(R^{(1)}_{\mu}) := \int |R^{(1)}_{\mu}|^2 \, dx dy$ is monotonically increasing in~$\mu$. In particular,
\begin{equation}
\label{eq:VK-pearl}
    \frac{d}{ d \mu} P(R^{(1)}_{\mu}) >0, \qquad \lim_{\mu \to \mu_{\rm lin}}P(R^{(1)}_{\mu})=0, \qquad
\lim_{\mu \to \infty} P(R^{(1)}_{\mu})= \Pcr,
\end{equation}
where
$$
\Pcr = \int_{{\Bbb R}^2} \left|R_{\mu, \rm 2D}^{(1), \rm free}\right|^2 \, dxdy
$$
is the critical power for collapse.
\item On an annular domain, the ground state solitary waves are radial  (ring-type) at low powers,
but undergo a {\em symmetry breaking} into a single-peak profile at a threshold power well below~$\Pcr$.
In particular, equation~\eqref{eq:dDim_R_ODE2} on the annulus,
\begin{enumerate}
  \item does not have a unique positive solution.
  \item has a positive solution which is not a ground state.
\end{enumerate}
   Similar results were obtained by Kirr et al.~\cite{Kirr-08} for the NLS-GP equation with an inhomogeneous (e.g., double-well) potential.
\end{enumerate}

Once we obtain the single-pearl profile, we construct the necklace solitary wave~$R^{(n)}_{\mu}$
from $n$~identical pearls, such that adjacent pearls have opposite signs. Therefore, by construction,
\begin{equation}
\label{eq:P=nP}
P(R^{(n)}_{\mu}) = n P(R^{(1)}_{\mu}).
\end{equation}
Hence, by~\eqref{eq:VK-pearl},
\begin{equation}
\label{eq:VK-necklace}
   \frac{d}{ d \mu}  P(R^{(n)}_{\mu}) > 0,
\qquad \lim_{\mu \to \mu_{\rm lin}}P(R^{(n)}_{\mu})=0, \qquad
\lim_{\mu \to \infty} P(R^{(n)}_{\mu})= n\Pcr.
\end{equation}

In Section~\ref{sec:Linear stability} we study the stability
of necklace solitary waves on bounded domains,  by (i)~computing the eigenvalues and eigenvectors of the associated linearized problem,
and (ii)~by solving the NLS with  an initial condition which is a perturbed necklace solitary wave.
These simulations show that:
\begin{enumerate}

\item Single-pearl (ground state) solitary waves $\psi_{\rm sw}^{(1)} = e^{i \mu z} R_\mu^{(1)}$ on bounded domains are stable for $\mu_{\rm lin} <\mu<\infty$, i.e., for $0<P<\Pcr$.

   \item  For $n \ge 2$, there exists~$\mu_{\rm c}$, where $\mu_{\rm lin} <\mu_{\rm c}<\infty$, such that the necklace solitary waves
     $\psi_{\rm sw}^{(n)} = e^{i \mu z} R_\mu^{(n)}$ are stable for $\mu_{\rm lin}< \mu <\mu_{\rm c}$,
  i.e., for
$$
0<P(R_{\mu}^{(n)})< P_{\rm th}^{\rm necklace} (n), \qquad P_{\rm th}^{\rm necklace} (n):=P(R_{\mu_{\rm cr}}^{(n)}),
$$
and unstable for $0<\mu-\mu_{\rm c} \ll 1$. 
 Thus, the low-power necklace solitary waves inherit the stability of the linear necklace modes $e^{i \mu_{\rm lin} z} Q^{(n)}$ from which they bifurcate.\footnote{Here $Q^{(n)}$ denotes a linear necklace solitary wave made out of $n$ identical $Q^{(1)}$ single pearls.}

\item    On convex domains such as rectangles and circles,
         necklace solitary waves are unstable  for $\mu_{\rm c}< \mu<\infty$, i.e., for
$$
P_{\rm th}^{\rm necklace} (n)<P(R_{\mu}^{(n)})<n \Pcr.
$$
       This may seem surprising, since~$R_\mu^{(n)}$  satisfies the Vakhitov-Kolokolov (VK) condition for stability
  $\frac{d}{ d \mu}  P(R_{\mu}) > 0$, see~\eqref{eq:VK-necklace}. We note, however, that the VK condition implies stability for ground-state solitary waves, which is not the case for necklace solitary waves. Indeed, the VK condition tests whether the solitary wave is susceptible to an amplitude (focusing) instability, whereby the solitary wave amplitude increases (decreases)
as its width decreases (increases). The VK condition, however, does not determine whether the
solitary wave is susceptible to other types of instabilities. See~\cite{Ilan-11,Sivan-08,Sivan-08A} for more details.

 \item  The threshold power for necklace instability is substantially smaller
than the critical power for collapse, i.e.,
$$
  P_{\rm th}^{\rm necklace}(n)
<  P_{\rm cr}.
$$
Thus, the necklace becomes unstable when the power of each pearl is below~$\Pcr/n$.
For example, the threshold power for instability of a necklace solitary wave with $n=4$ pearls is
$ \approx 0.55 \Pcr$ on a square domain
and $\approx 0.24 \Pcr$ on a circular domain,
i.e., when the power of each of the four pearls is roughly
$0.14 \Pcr$ and $0.06 \Pcr$, respectively.
As noted, single-pearl solitary waves are stable for $0<P<P_{\rm cr}$.
Hence, power-wise, {\em necklace solitary waves on circular and rectangular domains are considerably less stable than single-pearl solitary waves on these domains}.

\item  Necklace solitary waves on annular domains are also stable for $\mu_{\rm lin}< \mu <\mu_{\rm c}$ and unstable for $0<\mu-\mu_{\rm c} \ll 1$. 
In addition, $P_{\rm th}^{\rm necklace}\approx 0.24 \Pcr$ for a $4$-pearl
necklace solitary waves on an annular domain with a 1:2 ratio of the inner to outer radii, which is the same as for a circular necklace. Thus, the addition of a hole ``fails'' to increase the threshold power for necklace instability.

  Unlike on convex domains, however,
necklace solitary waves on the annulus have a second stability regime at powers well above~$\Pcr$.
For example, when the ratio of the inner to outer radii is 1:2,
4-pearl necklaces are stable for powers between $3.1\Pcr$ and $3.7\Pcr$.

 \item The instability of necklace solitary waves on circular, rectangular, and annular domains  share many similar features. Thus, the unstable modes are those which are symmetric with respect to one or more interfaces between pearls. These modes break the anti-symmetry between pearls, and thus allow for power transfer between pearls. Another common feature is that the unstable modes for a necklace with $2\times 2$~pearls are either
symmetric with respect to all four interfaces between pearls, or are symmetric in one direction
and anti-symmetric  in the perpendicular direction. In the latter case, the unstable modes
consist of two identical unstable modes of a 2-pearl necklace on half the domain.

 \item  In Section~\ref{sec:One-dimensional-necklaces} we consider the stability of necklace solitary waves of one-dimensional cubic and quintic NLS. We observe that the (in)stability properties of necklace solitary waves in the subcritical and critical cases are similar, in the sense that the instability occurs above a certain power threshold, and  is related to power transfer between pearls. In particular, this further shows that necklace instability is unrelated to collapse.

\end{enumerate}

In Section~\ref{sec:numerical} we discuss the numerical methods. In particular,
we introduce {\em a novel non-spectral variant of Petviashvili's renormalization method} that can be used to compute multi-dimensional solitary waves on bounded domains.

There are several unexpected results in this study:
\begin{enumerate}
  \item The threshold power for instability of a
necklace on a rectangular domain is {\em twice} that on a circular domain.
\item The threshold power for instability of a necklace on a circular domain is the same as
that of an annular domain. This is surprising, since the hole  reduces the interaction between the pearls, and so could be expected to have a stabilizing effect on the necklace.
\item In the case of necklace solitary waves on the annulus, there exists a second stability regime at powers
     above~$3\Pcr$.

\item To the best of our knowledge, our results for an annular domain provide the first example of a positive solitary wave of the {\em homogeneous} two-dimensional cubic NLS which
\begin{enumerate}
\item  satisfies the VK condition for stability, yet is unstable,
\item is not a ground state, and
\item is not the unique positive solitary wave.
\end{enumerate}
 As noted, such results were observed for the NLS-GP equation with an inhomogeneous potential~\cite{Kirr-08}.

\end{enumerate}

The possibility to propagate stable necklace solitary waves with powers well above~$\Pcr$ may be important for applications. While an annular domain does not correspond to a standard hollow-core fiber, the NLS on an annular domain models beam propagation in a
step-index optical fiber with a linear index of refraction
$$
n_0 = \left \{
\begin{array}{ccc}
n_0^{(1)}, \qquad \mbox{ \rm if} ~~0 \le r <r_1, \\
n_0^{(2)}, \qquad \mbox{ \rm if} ~~r_1 < r <r_2, \\
n_0^{(3)}, \qquad \mbox{ \rm if} ~~r_2 < r <r_3,
\end{array}
\right.
$$
where $0< r_1  <r_2  <r_3$ and $n_0^{(1)}, n_0^{(3)} <n_0^{(2)}$.

Our results can be extended in several directions. For example, NLS necklaces can be constructed on ``irregular'' two-dimensional domains, see e.g., Figure~\ref{fig:rectangle_step_necklace}. Such
necklace solitary waves may also be stable at powers above~$\Pcr$.
 One can also construct necklace solitary waves in dimensions higher than~2. For example, one can compute a single-pearl solitary wave
of the three-dimensional NLS on the unit cube, and then use it to construct a necklace solitary wave on a three-dimensional box. The question whether the NLS admits necklace solitary waves made out of non-identical pearls is open.

{\vspace{0mm} \parindent=0pt
         {\bf  A note on notations.} \hspace{0mm} \parindent=3ex}
 In this study we denote the profile of a necklace solitary wave with $n$~peaks (pearls) by~$R^{(n)}_{\mu}$.
In order to be consistent with this notation:
\begin{enumerate}
  \item The single-pearl profile,   i.e., the positive (ground state) solution
of~\eqref{eq:dDim_R_ODE2} is denoted by~$R^{(1)}_{\mu}$,
in contrast to the general convention of denoting the ground state by~$R^{(0)}_{\mu}$.

 \item The free-space necklace solutions in Section~\ref{sec:free-space} are parameterized by $n = 2m$.
\end{enumerate}

\section{Necklace solutions in free space}
   \label{sec:free-space}

The free space linear Schr\"odinger equation
\begin{equation}
   \label{eq:lin-Sch-necklace}
i \psi_z(z,x,y) + \Delta \psi = 0, \qquad -\infty<x,y<\infty, \quad z>0
\end{equation}
models the propagation of laser beams in a bulk linear medium.
This equation admits the
{\em Laguerre-Gaussian vortex solutions}
$$
   \psi_{\rm LG}^{\rm vortex}(z,r,\theta) := \frac{1}{L(z)}\left(\frac{r}{L(z)} \right)^m  e^{- \frac{1-4 i z }{L^2(z)} r^2+  i(m+1) \zeta(z)}
          e^{ i m \theta}, \qquad m=1,2, \dots,
$$
where  $L(z) = \sqrt{1+16 z^2}$,  $r = \sqrt{x^2+y^2}$, and $\zeta = -\arctan (4z)$. 
Since $\Delta \psi = \psi_{rr}+ \frac1r \psi_r+  \frac1{r^2} \psi_{\theta \theta}$,
if $\psi = A(z,r) e^{i m \theta}$ is a
solution of~\eqref{eq:lin-Sch-necklace}, then so is~$\psi = A(z,r) \cos (m \theta)$, as
in both cases $\psi_{\theta \theta} = -m^2 \psi$. Therefore,
$$
   \psi_{\rm LG}^{\rm necklace} := \frac{1}{L(z)}\left(\frac{r}{L(z)} \right)^m
 e^{- \frac{1-4 i z }{L^2(z)} r^2 + i(m+1) \zeta(z)}
          \cos( m \theta), \qquad m=1,2, \dots
$$
are also solutions of~\eqref{eq:lin-Sch-necklace}.
Whereas the amplitude of~$\psi_{\rm LG}^{\rm vortex}$ is radial, 
that of~$\psi_{\rm LG}^{\rm necklace}$ has a $|\cos( m \theta)|$ dependence. Hence,
$|\psi_{\rm LG}^{\rm necklace}|$ attains its maximum at the $n=2m$~points
$$
(r_j,\theta_j) = \left(r_m(z),\frac{j  \pi}{m}\right), \quad
r_m(z):=\arg \max_r  \left(\left(\frac{r}{L(z)} \right)^m e^{- \frac{r^2}{L^2(z)}} \right), \quad j=1, \dots, n.
$$
These $n$~peaks are located at equal distances along an expanding circle
of radius~$r_m(z)= r_m(0) L(z)$,  such that adjacent peaks have a $\pi$ phase difference.

The NLS~\eqref{eq:NLS_free_space} also admits vortex solutions of the form $\psi = A(z,r) e^{ i m \theta}$~\cite{Vortex-PD-08}.
In nonlinear propagation, however,
it is no longer true that if $\psi = A(z,r) e^{ i m \theta}$ is a solution,
then so is~$\psi = A(z,r) \cos (m \theta)$.
Nevertheless, the NLS admits solutions that have a ``necklace structure with $n$~pearls'':
\begin{lem}
  \label{lem:necklace}
  Let $\psi(z,r,\theta)$ be a solution of the {\rm NLS}~\eqref{eq:NLS_free_space}, let~$n$ be even, and let
\begin{equation}
  \label{eq:necklace_ic}
  \psi_0(r,\theta) = f(r) \cos \left( \frac{n}2 \theta\right).
\end{equation}
Then
\begin{enumerate}
  \item
$\psi$ is invariant under rotations by~$\frac{4 \pi}{n}$, i.e., $\psi(z,r,\theta) =\psi(z,r,\theta+\frac{4 \pi}{n})$.
 \item $\psi$ is antisymmetric with respect to the $n	$ rays
$\theta \equiv \Theta_j:=  \frac{\left(1+2j\right)\pi}{n}$,
 i.e.,
\begin{equation}
  \label{eq:antisymmetry-necklace}
     \psi(z,r,\Theta_j+\theta) =-\psi(z,r,\Theta_j+\theta),
 \qquad j = 1, \dots,n.
\end{equation}
\end{enumerate}
\end{lem}
\begin{proof}
This is a standard result, which follows from the uniqueness of NLS solutions.
\end{proof}

The repulsion between adjacent out-of-phase pearls (beams) can be understood as follows.
  By~\eqref{eq:antisymmetry-necklace},
\begin{equation}
 \label{eq:psi===0}
  \psi(z,r,\Theta_j) \equiv 0, \qquad j = 1, \dots, n.
\end{equation}
 Thus, as the necklace evolves, the electric field remains zero on the
rays that separate between adjacent beams
(see Figure~\ref{fig:segev_necklace_initial condition_8_pearls_no_green_lines}).
Hence, the dynamics of the $j$th pearl is governed by the NLS~\eqref{eq:NLS_bounded}
on the $j$th sector $D=\{\Theta_j<\theta<\tilde\Theta_{j+1}, 0<r<\infty\}$,
subject to Dirichlet boundary conditions on the rays $\theta\equiv \Theta_j$ and
$\theta \equiv \tilde\Theta_{j+1}$.
Because a Dirichlet boundary condition is reflecting,
there is no interaction between adjacent pearls (beams).
\begin{cor}
So long that it exists,
the solution of the {\rm NLS}~\eqref{eq:NLS_free_space} with the necklace initial condition~\eqref{eq:necklace_ic} maintains a necklace structure.
\end{cor}

 If the power of each pearl is sufficiently above the critical power for collapse~$\Pcr$, then
the solution collapses at a finite distance.\footnote{This follows from the variance identity on a bounded domain.} In that case, $\psi$~collapses simultaneously
at $n$~points.
If, however, the pearl power is below~$\Pcr$,
then $\psi$ scatters as $z \to \infty$. This scattering is slower than for a single pearl, however,
since the reflecting boundaries (or equivalently, the repulsion by the adjacent out-of-phase pearls)
slow down the pearl expansion.
In particular, if the power of each pearl is
slightly below~$\Pcr$, diffraction is almost balanced by the focusing Kerr nonlinearity
and the repulsion by the adjacent pearls. Consequently,
the pearl expansion is much slower than in the single-pearl case, as is demonstrated numerically 
in Figure~\ref{fig:segev_quarter_and_circle_necklace}.

\begin{figure}[ht!]
\begin{center}
\scalebox{.8}{\includegraphics{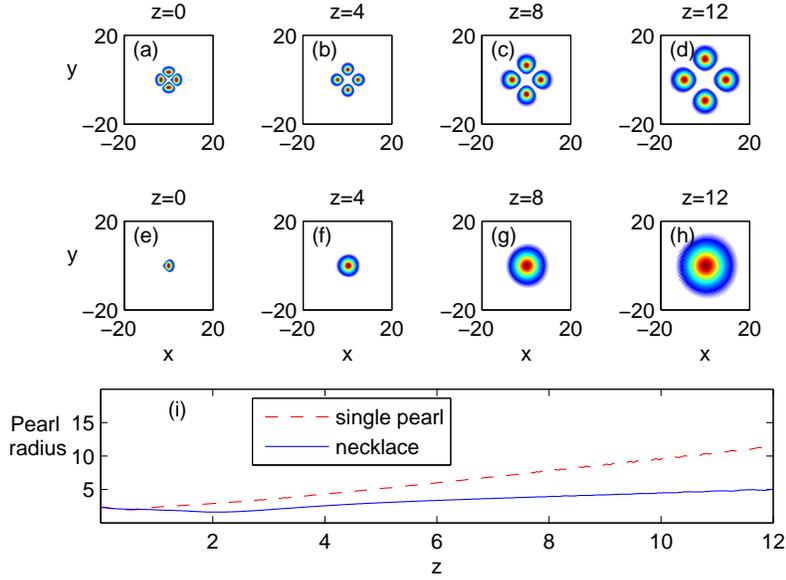}}
\caption{Solution of the NLS~\eqref{eq:NLS_free_space}.  (a)-(d): Contour plots of~$|\psi|$ at different distances,
for $\psi_{0}=\sqrt{2} \,\mbox{sech}(r-3.4)\cos(2 \theta)$. (e)-(h): The same  when~$\psi_0$ consists
of a single pearl, i.e., $\psi_{0}=\sqrt{2} \,\mbox{sech}(r-3.4)\cos(2 \theta)$ for $-\pi/4 \le \theta \le \pi/4$
and  $\psi_{0}\equiv 0$ otherwise.
 (i):~Radius of a single pearl as a function of~$z$,
for the 4-pearl necklace (solid), and for the single pearl (dashes).
}
\label{fig:segev_quarter_and_circle_necklace}
\end{center}
\end{figure}

Let us briefly discuss the stability of necklace beams in free space. By Lemma~\ref{lem:necklace},
as long as the initial condition is of
the form~\eqref{eq:necklace_ic}, the necklace structure is preserved.
Hence,  {\em  the necklace structure can only be destroyed by perturbations
that are not of the form~\eqref{eq:necklace_ic}}. Such ``azimuthal'' perturbations
break the antisymmetry with respect to the rays $\theta\equiv \Theta_j$,
and lead to power transfer between adjacent pearls.  Generally speaking, numerical
simulations suggest that necklace beams in free-space are more stable under random perturbations
than vortex beams~\cite{Gaeta-07,Soljacic-98}.

\subsection{No necklace solitary waves in free space}

   Ideally, a stable necklace beam should neither expand nor collapse, i.e., it should be a solitary wave.
   {\em While the free-space {\rm NLS} admits necklace-type solutions, see Lemma~{\rm \ref{lem:necklace}}, it does not admit necklace-type solitary waves.}
This is easy to see for the one-dimensional NLS\footnote{In one dimension we consider both the subcritical case $\sigma=1$ and the critical case $\sigma=2$.}
\begin{equation}
   \label{eq:NLS-1D-free}
i \psi_z(z,x) + \psi_{xx} + |\psi|^{2 \sigma} \psi = 0, \qquad -\infty<x<\infty, \quad z>0.
\end{equation}
Indeed, the solitary waves of~\eqref{eq:NLS-1D-free} are of the form
$\psi_{\rm sw} =e^{i\mu z}R_\mu(x)$, where~$R_\mu$ is a solution of
\begin{equation}
   \label{eq:R-1D-free}
  R''(x)-\mu R+|R|^{2 \sigma} R = 0, \qquad R(\pm \infty) = 0.
\end{equation}
The unique solution of~\eqref{eq:R-1D-free} is
$$
R_{\mu, \rm 1D}^{(1), \rm free}(x) = \mu^{\frac{1}{2 \sigma}}R_{\rm 1D}^{(1), \rm free}(\mu^{\frac{1}{2}}x),
\qquad
R_{\rm 1D}^{(1), \rm free}(x) := (1+\sigma)^{\frac{1}{2\sigma}} \, \mbox{\rm sech}^{\frac{1}{\sigma}} (\sigma x).
$$
Since~$R_{\mu, \rm 1D}^{(1), \rm free}$ has a single peak,
there are no necklace solutions of~\eqref{eq:R-1D-free}.
In two dimensions, we have
\begin{lem}
The free-space two-dimensional {\rm NLS}~\eqref{eq:R-2D-free} does not admit necklace-type solitary waves.
\end{lem}
\begin{proof}
We provide an informal proof.
 Assume by negation that there is a necklace solution with e.g., $n=4$~pearls. Without loss of generality, their peaks are located at~$(\pm a,\pm a)$.
By antisymmetry, $R\equiv 0$ on the $x$-axis and $y$-axis. Therefore, there exists a nontrivial, positive,
single-peak solution of~\eqref{eq:dDim_R_ODE2} on the positive quarter-plane $S_\infty := \{0<x,y<\infty\}$,
whose peak is at~$(a, a)$.
By continuity, this solution is the limit of
positive, single-peak solutions of~\eqref{eq:dDim_R_ODE2} on the square $S_A := \{0 <x,y<A\}$.
By symmetry, however, the peak of these solutions is at~$(A/2,A/2)$.\footnote{This statement is further supported by
Corollary~\ref{cor:square}.} Since
$\lim_{A \to \infty}(A/2,A/2) \not= ( a,a)$, we reach a contradiction.
\end{proof}

\section{NLS on a bounded domain}
\label{sec:review}

In this section we briefly review the NLS on bounded domain. See~\cite{bounded-01} and~\cite[Chapter~16]{NLS-book}
for further details.
Solutions of~(\ref{eq:NLS_bounded}) conserve their {\em power} and Hamiltonian, i.e.,
$$
P(\psi) \equiv P(\psi_0), \qquad H(\psi) \equiv H(\psi_0) ,
$$
where
$$
P(\psi):= \int_{D} |\psi|^2 \, dx dy, \qquad H(\psi):= \int_{D} |\nabla \psi|^2 \, dx dy- \frac12\int_{D} |\psi|^4 \, dx dy.
$$

  Recall that the free-space {\em ground-state} solitary
waves\footnote{i.e., the positive solutions of~\eqref{eq:R-2D-free}.} depend on~$\mu$ through the scaling
$$
R^{(1), \rm free}_{\mu , \rm 2D}({\bf x}) = \mu^{\frac{1}{2}}R^{(1), \rm free}_{\rm 2D}(\mu^{\frac{1}{2}}{\bf x}),
\qquad R^{(1), \rm free}_{\rm 2D} := R^{(1), \rm free}_{\mu=1 , \rm 2D}.
$$
Consequently, their power is independent of~$\mu$, i.e.,
\begin{equation}
  \label{eq:P(mu)=P}
P(R^{(1), \rm free}_{\mu , \rm 2D}) \equiv P(R^{(1), \rm free}_{\rm 2D}) = \Pcr,
\end{equation}
where $\Pcr$ is the critical power for collapse.
On bounded domains, however,  there is no such scaling for~$R_\mu^{(1)}$.
Indeed, numerical simulations and some analytic results show that the
power of $R^{(1)}_{\mu}$ is strictly increasing in~$\mu$.

 The absence of a scaling invariance
can be used to give the following variational characterization
of the ground-state solitary waves  on bounded domains:
\begin{conj}
\label{conj:unique_min}
Let $D$ be a regular bounded domain, and let $\mu_{\rm lin}$ be the first eigenvalue of~\eqref{eq:lin_R_mu_in_D}.
 Then
for all $\mu \in (\mu_{\rm lin},  \infty)$,
\begin{enumerate}
  \item Eq.~\eqref{eq:dDim_R_ODE2} has minimal-power nontrivial solution, denoted by~$R^{(1)}_{\mu}$, which is positive  and unique (up to symmetries that leave the domain invariant).
   \item
$R^{(1)}_{\mu}$
is the unique real minimizer over all  $U(x,y) \in H_0^1(D)$
 of
$$
  \inf_{||U||_2^2 =||R^{(1)}_{\mu}||_2^2} H(U) .
$$
\item
 For all $0<P<\Pcr$,  there exists a unique
$\mu_P \in (\mu_0,  \infty)$ such that
$R^{(1)}_{\mu_P}$ is the unique real minimizer over all  $U(x,y) \in H_0^1(D)$ of
\begin{equation}
   \label{eq:variational}
  \inf_{||U||_2^2 =P} H(U) .
\end{equation}
\end{enumerate}
\end{conj}

Parts of Conjecture~\ref{conj:unique_min} were proved by Fibich and Merle~\cite{bounded-01}:
The existence of a minimizer of~\eqref{eq:variational} was proved for any regular bounded domain, and
its uniqueness was proved for $D = B_1$.\footnote{Since then  one can use Steiner symmetrization to conclude that the minimizer is radial.} Since $H(|u|) \le H(u)$ for $u \in H^1_0$, the minimizer can be assumed to be non-negative.

\begin{cor}
   \label{cor:square}
Assume that Conjecture~{\rm \ref{conj:unique_min}}  holds for $D = [-1, 1]^2$.
Then the maximum of~$R^{(1)}_{\mu}(x,y)$ is attained at the origin.
\end{cor}
\begin{proof}
We provide a sketch of a proof.  Assume by negation that the maximum of~$R^{(1)}_{\mu}$
is not attained at the origin.
If we apply Steiner symmetrization to~$R^{(1)}_{\mu}(x,y)$ in~$x$ and then in~$y$, we
obtain a function~$S_\mu(x,y)$ which is symmetric in~$x$ and in~$y$, and is monotonically decreasing in~$|x|$ and
in~$|y|$. Therefore, its peak is
at the origin. Furthermore,
$$
\|S_\mu\|_2 = \|R^{(1)}_{\mu}\|_2, \quad \|S_\mu\|_4 = \|R^{(1)}_{\mu}\|_4, \quad
\|\nabla S_\mu\|_2 \le \|\nabla R^{(1)}_{\mu}\|_2.
$$
In fact, we claim that $\|\nabla S_\mu\|_2 < \|\nabla R^{(1)}_{\mu}\|_2$.
Therefore, $S_\mu \in H_0^1(D)$,  $\|S_\mu\|_2 = \|R^{(1)}_{\mu}\|_2$
and $H(S_\mu)<H( R^{(1)}_{\mu})$, which is in contradiction to our assumption that~$R^{(1)}_{\mu}$ is a minimizer.

Indeed, if $\|\nabla S_\mu\|_2  =  \|\nabla R^{(1)}_{\mu}\|_2$, then $\nabla R^{(1)}_{\mu}$ vanishes on a set of
positive measure inside~$D$ (see e.g., \cite[Thoerem~2.6]{Fusco}). In that case, however, by elliptic regularity,
the unique continuation of $R^{(1)}_{\mu}$ from the set of positive measure where
$\nabla R^{(1)}_{\mu}$ vanishes to~$D$  is $R^{(1)}_{\mu}\equiv 0$.
\end{proof}

%

\section{Necklace solitary waves}
\label{sec:Solitary wave}

In this section we construct necklace solitary waves on various bounded domains.
The numerical method used for computing these solutions is discussed in Section~\ref{sec:Computation_of_the_solitary_waves}.

\subsection{Rectangular necklaces}
   \label{sec:2d_rectangle_solitary_wave}

To construct a necklace solitary wave on a rectangular domain, we first consider a single pearl on a square.

 \subsubsection{Single pearl on a square}

 Let
$D$ be the square $[-1, 1]^2$. Then
equation (\ref{eq:dDim_R_ODE2}) reads
\begin{subequations}
   \label{eq:Rlambdarectangular}
\begin{equation}
  \label{eq:Rlambdarectangular_a}
\Delta R(x,y)-\mu R+|R|^{2}R=0,\qquad -1<x,y<1,
\end{equation}
\begin{equation}
  \label{eq:Rlambdarectangular_b}
R(x=\pm1,y)=R(x,y=\pm1)=0.
\end{equation}
Since we look  for a positive solution, we add the requirement that
\begin{equation}
 R(x,y)>0,\qquad -1<x,y<1.
\end{equation}
\end{subequations}
Following Fibich and Merle~\cite{bounded-01}, the solutions of~\eqref{eq:Rlambdarectangular},
which we denote by~$R^{(1)}_{\mu}$,  bifurcate from
\begin{equation}
   \label{eq:Rlambdarectangular_b-lin}
Q^{(1)}(x,y) := \cos \left(\frac{\pi x}{2}\right)\cos\left(\frac{\pi y}{2}\right), \qquad \mu_{\rm lin} =
-\frac{\pi^{2}}{2},
\end{equation}
which is the ground-state eigenpair of~\eqref{eq:lin_R_mu_in_D} with $D=[-1, 1]^2$.

 As~$\mu$
increases from~$\mu_{\rm lin}$, $R^{(1)}_{\mu}$ becomes
more localized as its power increases and the nonlinearity becomes more pronounced, see Fig.~\ref{fig:2dRxy}. In particular, as $\mu
\to \infty$, $R^{(1)}_{\mu}$~ ``does not
feel'' the reflecting boundary, and so it
approaches the corresponding free-space solitary wave,
i.e., $ R^{(1)}_{\mu}\thicksim  R^{(1), \rm free}_{\mu , \rm 2D}$, where
$R^{(1), \rm free}_{\mu , \rm 2D}$ is the positive solution of~\eqref{eq:R-2D-free}.
Hence,
 $\lim_{\mu\rightarrow\infty} P(R_{\mu}^{(1)})=P_{\rm cr}$, where $P_{\rm cr} = \int |R^{(1), \rm free}_{\rm 2D}|^2 \, dxdy$ is the critical power for collapse~\cite{critical-00}.

\begin{figure}[ht!]
\begin{center}
\scalebox{0.75}{\includegraphics{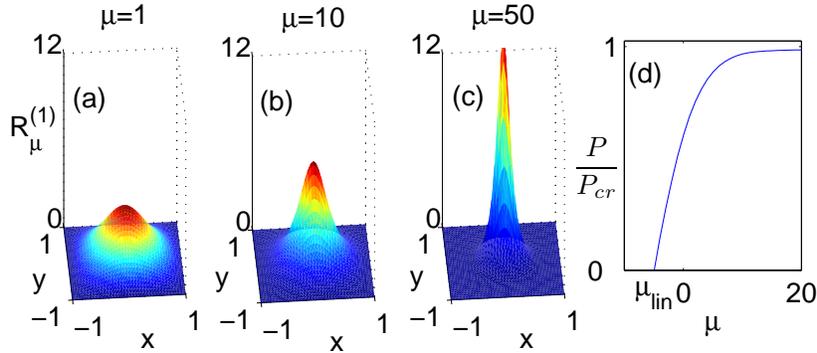}}
\caption{ A single pearl on a square, i.e., the solution of~(\ref{eq:Rlambdarectangular}).
 (a)~$\mu=1$. (b)~$\mu=10$. (c)~$\mu=50$. (d)~The pearl power $P(R^{(1)}_{\mu})=\int_{-1}^{1}\int_{-1}^{1}|R^{(1)}_{\mu}|^{2}dx dy$ as a function of~$\mu$. Here, $\mu_{\rm lin} \approx -5$.
  }
   \label{fig:2dRxy}
\end{center}
\end{figure}

\subsubsection{Rectangular necklace solitary waves}

We can use the single-pearl solution~$R^{(1)}_{\mu}$ to construct a rectangular necklace solitary wave with
$l\times k$~pearls, by letting adjacent pearls have opposite phases (signs).
For example, Fig.~\ref{fig:rectangle_necklace} shows a rectangular necklace solitary
wave with $2 \times 3$ pearls.

\begin{lem}
 \label{lemma:R_lm}
 $R^{(l \times k)}_{\mu}$ is a rectangular necklace solution of equation \eqref{eq:dDim_R_ODE2} with $l \times k$ pearls.
\end{lem}
\Beginproof
 Clearly, $R^{(l \times k)}_{\mu}$ is a smooth solution of~(\ref{eq:dDim_R_ODE2a}) inside each "cell", and it satisfies the boundary conditions~(\ref{eq:dDim_R_ODE2b}). To show that $R^{(l \times k)}_{\mu}$ satisfies equation~\eqref{eq:dDim_R_ODE2a} at the interfaces between pearls, it is enough to show that it is twice continuously
differentiable there. Clearly,
 $R^{(l \times k)}_{\mu}$ is continuous at the interfaces, since the pearls vanish there.
Consider, for example, the interface $x\equiv 1$ in Fig.~\ref{fig:rectangle_necklace}(a).
 Note that if $R^{(1)}_{\mu}(x,y)$ is a solution of~\eqref{eq:Rlambdarectangular}, then so is $R^{(1)}_{\mu}(-x,y)$.
 Therefore,  $R^{(1)}_{\mu}(-x,y)=R^{(1)}_{\mu}(x,y)$, and so
\begin{subequations}
    \label{eq:C2}
\begin{equation}
-\frac{\partial}{\partial x}R^{(l \times k)}_{\mu}(x=-1+,y)=\frac{\partial}{\partial x}R^{(l \times k)}_{\mu}(x=1-,y).
\end{equation}
In addition, by construction of~$R^{(l \times k)}_{\mu}$, we have that  $R^{(l \times k)}_{\mu}(x,y)=-R^{(l \times k)}_{\mu}(x-2,y)$. Hence,
\begin{equation}
\frac{\partial}{\partial x}R^{(l \times k)}_{\mu}(x=1+,y)=-\frac{\partial}{\partial x}R^{(l \times k)}_{\mu}(x=-1+,y).
\end{equation}
\end{subequations}
 Therefore, by~\eqref{eq:C2}, $\frac{\partial}{\partial x}R^{(l \times k)}_{\mu}$ is continuous at~$x\equiv1$.  Since $R^{(l \times k)}_{\mu}\equiv0$ on the interface $x\equiv1$,
then $\frac{\partial}{\partial y}R^{(l \times k)}_{\mu} \equiv0$,
$\frac{\partial^{2}}{\partial x\partial y}R^{(l \times k)}_{\mu} \equiv0$,   and
$\frac{\partial^{2}}{\partial y^{2}}R^{(l \times k)}_{\mu} \equiv0$ there. Therefore,
 $\frac{\partial^{2}}{\partial x^{2}} R^{(l \times k)}_{\mu}
=(\mu-|R^{(l \times k)}_{\mu}|^{2})R^{(l \times k)}_{\mu}-\frac{\partial^{2}}{\partial y^{2}} R^{(l \times k)}_{\mu} \equiv0$ there. Hence, $\frac{\partial}{\partial y}R^{(l \times k)}_{\mu}$,
$\frac{\partial^{2}}{\partial x\partial y}R^{(l \times k)}_{\mu}$,
$\frac{\partial^{2}}{\partial y^{2}}R^{(l \times k)}_{\mu}$ and $\frac{\partial^{2}}{\partial x^{2}}R^{(l \times k)}_{\mu}$
are also continuous at~$x\equiv1$.
\Endproof

 \Remark Note that  $R^{(l \times k)}_{\mu}$  is identically zero along the lines separating adjacent pearls,
and is antisymmetric with respect to these lines.

  \begin{figure}[ht!]
\begin{center}
\scalebox{0.6}{\includegraphics{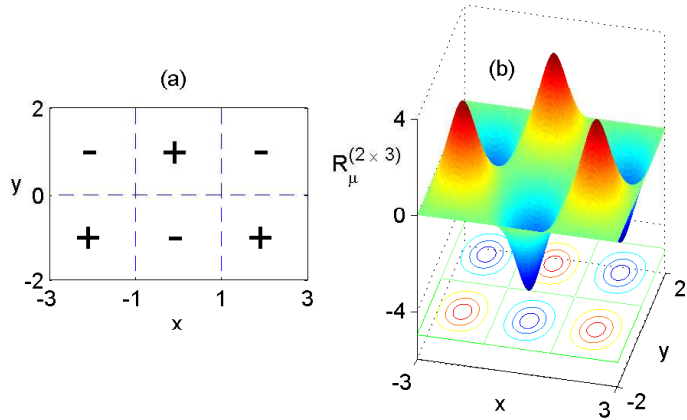}}
 \caption{(a)~Construction of a rectangular necklace solitary wave with $2 \times 3$ pearls. The  $'\pm'$ symbols correspond to~$\pm R^{(1)}_{\mu}$. Note that $R^{(2 \times 3)}_{\mu}$ is identically zero along the dashed lines, and is antisymmetric with respect to these lines. (b)~The rectangular necklace solitary wave~$R_{\mu=1}^{(2 \times 3)}$.
 }
\label{fig:rectangle_necklace}
\end{center}
\end{figure}

\Remark  The single-pearl solution can also be used to construct necklace solitary waves on non-rectangular domains,
   see e.g., Figure~\ref{fig:rectangle_step_necklace}.

  \begin{figure}[ht!]
\begin{center}
\scalebox{0.6}{\includegraphics{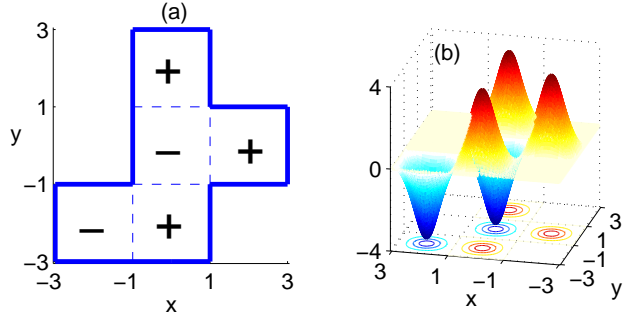}}
 \caption{Same as Fig.~\ref{fig:rectangle_necklace} for a non-rectangular domain with 5 pearls.
     }\label{fig:rectangle_step_necklace}
\end{center}
\end{figure}

\subsection{Circular necklaces}
  \label{sec:2d_quarter_circle_solitary_wave}

 To construct a necklace solitary wave on a circular domain, we repeat the procedure for
a rectangular domain (Section~\ref{sec:2d_rectangle_solitary_wave}).
Thus, we first compute a single pearl solution on a sector of a circle, and then use it to construct
a circular necklace solitary wave.

 \subsubsection{Single pearl on a sector of a circle}

Let~$D=\{r\in[0,1],\theta\in[0,\frac{2\pi}{n}]\}$ be the sector of the unit circle. Then
 equation (\ref{eq:dDim_R_ODE2}) reads
\begin{subequations}
   \label{eq:2Dim_R_ODE_general}
\begin{equation}
R_{rr}(r,\theta) +\frac{1}{r}R_{r}+\frac{1}{r^{2}}R_{\theta\theta}-\mu
R+|R|^{2}R=0,\qquad 0<r<1,\quad 0<\theta<\frac{2\pi}{n},
\end{equation}
\begin{equation}
R(r=1,\theta)=R(r,\theta=0)=R\left(r,\theta=\frac{2\pi}{n}\right)=0.
\end{equation}
Since we look  for a positive solution, we add the requirement that
\begin{equation}
 R(r,\theta)>0,\qquad 0<r<1,\quad 0<\theta<\frac{2\pi}{n}.
\end{equation}
\end{subequations}
Solutions of~\eqref{eq:2Dim_R_ODE_general} bifurcate from the positive eigenfunction
 of~\eqref{eq:lin_R_mu_in_D} with $D=\{r\in[0,1],\theta\in[0,\frac{2\pi}{n}]\}$,
 which is given by
\begin{equation}
  \label{eq:2Dim_R_ODE_general-lin}
 Q^{(1)}: = J_{\frac{n}{2}}\left(k_{\frac{n}{2}}
r\right)\sin \left( \frac{n}{2}\theta \right), \qquad \mu_{\rm lin} = -k_{\frac{n}{2}}^2,
\end{equation}
where $J_{\frac{n}{2}}$ is the Bessel function of
order~$\frac{n}{2}$ of the first kind, and $k_{\frac{n}{2}}$ is the first positive
root of~$J_{\frac{n}{2}}$.
As~$\mu$ increases from~$\mu_{\rm lin}$, $R^{(1)}_{\mu}$ becomes more localized as its power increases, see
Fig.~\ref{fig:quarter_circle}.
In particular,
$ R^{(1)}_{\mu} \thicksim R^{(1), \rm free}_{\mu, \rm 2D}$
 as $\mu \to \infty$.
Consequently, $\lim_{\mu\rightarrow\infty} P(R_{\mu}^{(1)})=P_{\rm cr}$.

\begin{figure}[ht!]
\begin{center}
\scalebox{0.75}{\includegraphics{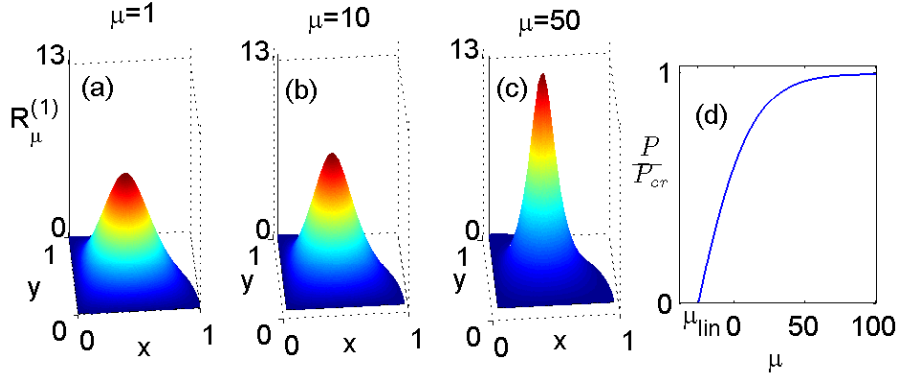}}
\caption{A single pearl on a quarter of the unit circle, i.e., the ground-state solution of~(\ref{eq:2Dim_R_ODE_general})
with $n=4$.
 (a)~$\mu=1$. (b)~$\mu=10$. (c)~$\mu=50$. (d)~The pearl power $P(R^{(1)}_{\mu})=\int_{0}^{\frac{\pi}{2}}\int_{0}^{1}r|R^{(1)}_{\mu}|^{2}dr d\theta$ as a function of~$\mu$.
Here, $\mu_{\rm lin}\thickapprox-26.4$.
 }
 \label{fig:quarter_circle}
\end{center}
\end{figure}


\subsubsection{Circular necklace solitary waves}
\label{sec:Circular necklace solitary waves}

When~$n$ is even,  we can use the single-pearl solution~$R^{(1)}_{\mu}$
to construct a circular necklace solitary wave with $n$~pearls, by letting neighboring pearls have opposite signs. For example, Fig.~\ref{fig:four_quarters}
shows a circular necklace solitary wave with 4 pearls.
  Note that $R^{(n)}_{\mu}$ is identically zero along the rays that are half way through between adjacent pearls, and is antisymmetric with respect to these rays.

  \begin{lem}
Let $n$ be even. Then $R^{(n)}_{\mu}$ is a circular necklace solution of equation~\eqref{eq:dDim_R_ODE2} with $n$~pearls.
\end{lem}
\begin{proof}
The proof is the same as that of Lemma~\ref{lemma:R_lm}.
\end{proof}

\begin{figure}[ht!]
\begin{center}
\scalebox{0.6}{\includegraphics{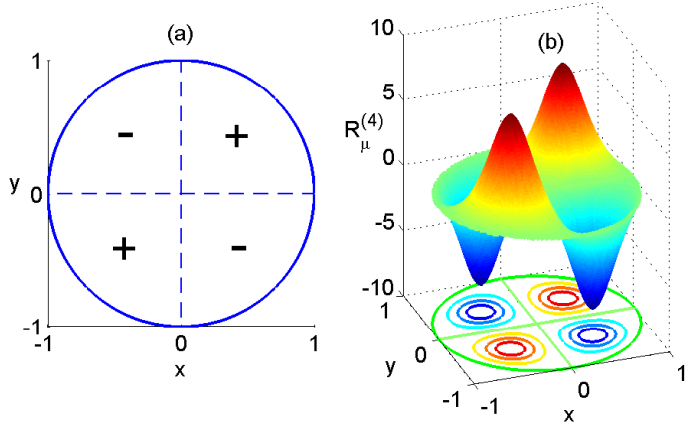}} \caption{a)~Construction of the circular necklace solitary wave with $4$~pearls.
 The symbols $'\pm'$ correspond to $\pm R^{(1)}_{\mu}$. Note that $R^{(4)}_{\mu}$ is identically zero along the dashed lines, and is antisymmetric with respect to these lines. (b)~The circular necklace solitary wave $R_{\mu=1}^{(4)}$.
}
 \label{fig:four_quarters}
\end{center}
\end{figure}

\subsection{Annular necklaces}
\label{sec:2d_quarter_circle_with_hole_solitary_wave}

 To construct a necklace solitary wave on an annular domain, we first consider a single pearl on a sector of an annulus, and then use it to construct necklace solutions on the whole annular domain.

 \subsubsection{Single pearl on a sector of an annulus}

Let~$D= \{r\in[r_{\rm min},r_{\rm max}],\theta\in[0,\frac{2\pi}{n}]\}$ be the sector of the annulus.
Then equation~(\ref{eq:dDim_R_ODE2}) reads
\begin{subequations}
  \label{eq:2Dim_R_ODE_hole}
\begin{equation}
 R_{rr}(r,\theta)+\frac{1}{r}R_{r}+\frac{1}{r^{2}}R_{\theta\theta}-\mu
R+|R|^{2}R=0,\qquad r_{\rm min}<r<r_{\rm max},\quad 0<\theta<\frac{2\pi}{n},
\end{equation}
    \label{eq:2Dim_R_ODE_hole_a}
\begin{equation}
R(r_{\rm min},\theta)=R(r_{\rm max},\theta)=R(r,\theta=0)=R\left(r,\theta=\frac{2\pi}{n}\right)=0.
\end{equation}
Since we look  for a positive solution, we add the requirement that
\begin{equation}
 R(r,\theta)>0,\qquad r_{\rm min}<r<r_{\rm max},\quad 0<\theta<\frac{2\pi}{n}.
\end{equation}
\end{subequations}
Solutions of~\eqref{eq:2Dim_R_ODE_hole} bifurcate from the ground-state eigenfunction
 of~\eqref{eq:lin_R_mu_in_D} with $D=\{r\in[r_{\rm min},r_{\rm max}],\theta\in[0,\frac{2\pi}{n}]\}$,
 which is given by
$$
 Q^{(1)} = \left[J_{\frac{n}{2}}\left( k_{\frac{n}{2}}r \right)-\frac{J_{\frac{n}{2}}\left( k_{\frac{n}{2}} r_{\rm min}\right)}{Y_{\frac{n}{2}}\left( k_{\frac{n}{2}} r_{\rm min}\right)}Y_{\frac{n}{2}}\left( k_{\frac{n}{2}}r \right)\right] \sin \left(\frac{n\theta}{2}\right),
  \qquad \mu_{\rm lin} = -k_{\frac{n}{2}}^2,
$$
where~$J_{\frac{n}{2}}$ and~$Y_{\frac{n}{2}}$ are Bessel functions of
order~$\frac{n}{2}$ of the first and second kind, respectively,
 and $k_{\frac{n}{2}}$ is the smallest positive root of
\begin{equation}
  \label{eq:k_n/2}
J_{\frac{n}{2}}\left( k_{\frac{n}{2}} r_{\rm max}\right)Y_{\frac{n}{2}}\left( k_{\frac{n}{2}} r_{\rm min}\right)-Y_{\frac{n}{2}}\left( k_{\frac{n}{2}}
r_{\rm max} \right)J_{\frac{n}{2}}\left( k_{\frac{n}{2}} r_{\rm min}\right)=0.
\end{equation}
 As~$\mu$ increases from~$\mu_{\rm lin}$, $R^{(1)}_{\mu}$ becomes more localized as its power increases, see
Fig.~\ref{fig:quarter_circle_hole}. In particular,
$R^{(1)}_{\mu} \thicksim R^{(1), \rm free}_{\mu, \rm 2D}$ as $\mu \to \infty$.
Consequently, $\lim_{\mu \to \infty} P(R_{\mu}^{(1)}) =  P_{\rm cr}$.

\begin{figure}[ht!]
\begin{center}
\scalebox{0.75}{\includegraphics{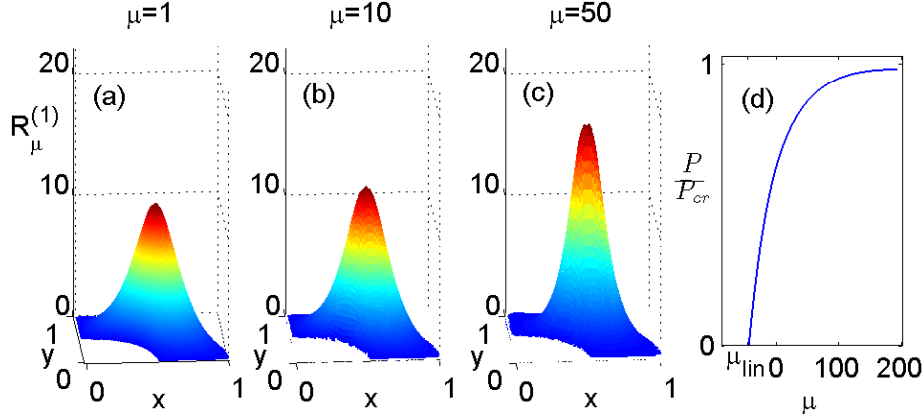}}
 \caption{
 A single pearl on a quarter annulus, i.e., the solution of~(\ref{eq:2Dim_R_ODE_hole}) with~$n=4$.
a)~$\mu=1$. (b)~$\mu=10$. (c)~$\mu=50$. (d)~The pearl power
$P(R^{(1)}_{\mu})=\int_{0}^{\frac{2\pi}{n}} d\theta \int_{r_{\rm min}}^{r_{\rm max}}r|R^{(1)}_{\mu}|^{2}\, dr$ as a function of~$\mu$.
 Here $r_{\rm min}=0.5$, $r_{\rm max}=1$, and~$\mu_{\rm lin}\thickapprox-45.5$.
}
 \label{fig:quarter_circle_hole}
\end{center}
\end{figure}

\subsubsection{Annular necklace solitary waves}
\label{sec:Annular necklace solitary waves}

We can use $R^{(1)}_{\mu}$ as a "building block" to construct an annular necklace solitary wave with $n$~pearls,
by letting neighboring pearls have opposite signs. For example, Fig.~\ref{fig:four_quarters_hole} shows an annular necklace solitary wave with $4$~pearls.


  \begin{lem}
Let $n$ be even. Then $R^{(n)}_{\mu}$ is an annular necklace solution
of equation~\eqref{eq:dDim_R_ODE2} with $n$~pearls.
\end{lem}
\begin{proof}
The proof is the same as that of Lemma~\ref{lemma:R_lm}.
\end{proof}

\begin{figure}[ht!]
\begin{center}
\scalebox{0.6}{\includegraphics{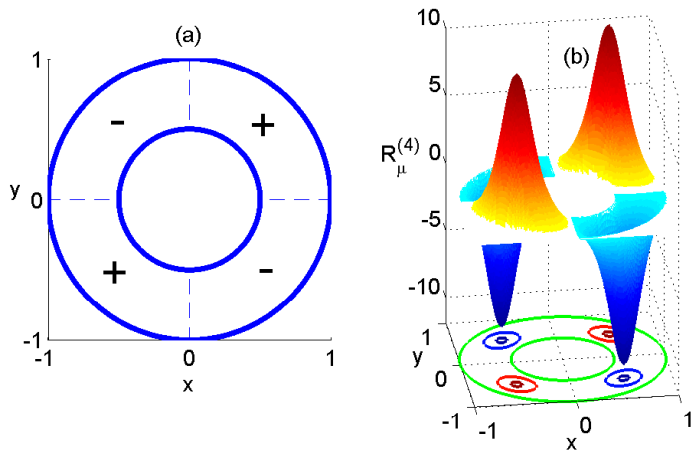}}
\caption{
(a) Construction of an annular necklace solitary wave with $4$~pearls. The symbols $'\pm'$ correspond to $\pm R^{(1)}_{\mu}$. Note that $R^{(4)}_{\mu}$ is identically zero along the dashed lines, and is antisymmetric with respect to these lines. (b)~The annular necklace solitary wave $R_{\mu=1}^{(4)}$. Here $r_{\rm min}=0.5$ and $r_{\rm max}=1$.
}
 \label{fig:four_quarters_hole}
\end{center}
\end{figure}

\subsubsection{Annular ground-state solitary waves  (symmetry breaking)}
\label{sec:Annular-radial-solitary-waves}

We can look for positive radial solitary waves $\psi_{\rm sw} = e^{ i \mu z} R^{\rm ring}_{\mu}(r)$ of the NLS~\eqref{eq:NLS_bounded} on the annulus.
The profile of these {\em ring-type} solitary waves satisfies
\begin{subequations}
  \label{eq:2Dim_R_ODE_hole_radial}
\begin{equation}
 R_{rr}(r)+\frac{1}{r}R_{r}-\mu
R+|R|^{2}R=0,\qquad r_{\rm min}<r<r_{\rm max},
\end{equation}
\begin{equation}
R(r_{\rm min})=R(r_{\rm max})=0,
\end{equation}
\begin{equation}
 R(r)>0,\qquad r_{\rm min}<r<r_{\rm max}.
\end{equation}
\end{subequations}
Because of the hole and the requirement of radial symmetry, these solitary waves {\em cannot} approach the
free-space two-dimensional ground-state~$R^{(1), \rm free}_{\mu, \rm 2D}$
 as $\mu \to \infty$. 
Rather,
\begin{equation}
   \label{eq:R^ring}
R^{\rm ring}_{\mu} \sim \sqrt{\mu} R^{(1), \rm free}_{\rm 1D}(\sqrt{\mu}(r-r_{{\rm M}, \mu})),
\qquad \mu \to \infty,
\end{equation}  where
$R^{(1), \rm free}_{\rm 1D}$ is the positive solution of~\eqref{eq:R-1D-free}
and $r_{{\rm M}, \mu}:=\arg\max R^{\rm ring}_{\mu}(r)$.
Hence,
\begin{eqnarray*}
 P(R^{\rm ring}_{\mu})= 2 \pi \int_{r_{\min}}^{r_{\max}} |R^{\rm ring}_{\mu}|^2 \, rdr
  \sim
 2 \pi \int_{r_{\min}}^{r_{\max}} |\sqrt{\mu} R^{(1), \rm free}_{\rm 1D}(\sqrt{\mu}(r-r_{{\rm M}, \mu}))|^2 \, rdr
\sim  c\sqrt{\mu},
\end{eqnarray*}
where $c= 2 \pi r_{{\rm M}, \mu}   \int_{-\infty}^{\infty} |R^{(1), \rm free}_{\rm 1D}(x)|^2 \, dx$.
Consequently, $\lim_{\mu \to \infty} P(R^{\rm ring}_{\mu}) =  \infty$.

In Figure~\ref{fig:R_full_annular_and_P_as_mu}
we compute the ground-state solutions of~\eqref{eq:dDim_R_ODE2} on the annulus $\{0.5<|{\bf x}|<1\}$
using the non-spectral renormalization method (Section~\ref{sec:Computation_of_the_solitary_waves})
{\em without imposing radial symmetry}.  These positive solutions are radial
for $\mu_{\rm lin}<\mu< \mu_{\rm c}$,\footnote{i.e., $R^{(1)}_\mu =  R^{\rm ring}_{\mu}$ for $\mu_{\rm lin}<\mu< \mu_{\rm c}$.} where  $\mu_{\rm lin} \approx -39$ and
$\mu_{\rm c}  \approx -38.2$. For $ \mu_{\rm c}<\mu<\infty$,
they become non-radial. In particular, 
$R^{(1)}_{\mu} \sim \sqrt{\mu} R^{(2), \rm free}_{\rm 1D}(\sqrt{\mu}({\bf x}-{\bf x}_{{\rm M}, \mu}))$
as $\mu \to \infty$.\footnote{Non-radial ground states on the annulus are only determined up to a rotation about the origin.  Therefore, ${\bf x}_{{\rm M}, \mu} = r_{{\rm M}, \mu}(\cos \theta, \sin \theta)$, where $\theta$ is arbitrary.
}
 Consequently, $\lim_{\mu \to \infty} P(R_{\mu}^{(1)}) =  P_{\rm cr}$.

The symmetry-breaking  at~$\mu_{\rm c}$ is related to the variational characterization of the ground-state
(see Conjecture~\ref{conj:unique_min}). Indeed, by~\eqref{eq:R^ring}, $H(R^{\rm ring}_{\mu}) \sim   \mu^{3/2}$ as $\mu \to \infty$. Therefore, high-power solutions of~\eqref{eq:dDim_R_ODE2} on $D=\{r_{\min}<|{\bf x}|<r_{\max}\}$ can achieve a smaller Hamiltonian by ``adopting'' a non-radial profile.

Note that  for $ \mu_{\rm c}<\mu<\infty$,  equation~\eqref{eq:dDim_R_ODE2} on the annulus
admits both a  radial solution and a non-radial one.
{\em To the best of our knowledge, this is the first example that equation~\eqref{eq:dDim_R_ODE2}
admits two different positive solutions, and that a positive solution of~\eqref{eq:dDim_R_ODE2}
is not a ground state.} As noted above, such results were observed for the inhomogeneous NLS~\cite{Kirr-08}.

\begin{figure}[ht!]
\begin{center}
\scalebox{0.8}{\includegraphics{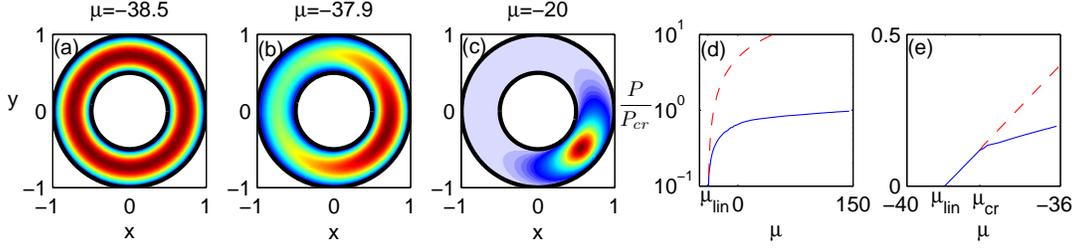}}
 \caption{Ground-state solutions of~\eqref{eq:dDim_R_ODE2} on the annulus $\{r_{\rm min}<|{\bf x}|<r_{\rm max}\}$,
where $r_{\rm min}=0.5$ and~$r_{\rm max}=1$.   Here $\mu_{\rm lin} = -38.2$.
 (a)~$\mu=-38.5$. (b)~$\mu=-37.9$. (c)~$\mu=-20$.
(d)~$P(R^{(1)}_{\mu})$ as a function of~$\mu$ (solid blue line).
The dashed red line is $P(R^{\rm ring}_{\mu})$.
e)~Zoom in near the bifurcation at~$\mu_c$.
}
 \label{fig:R_full_annular_and_P_as_mu}
\end{center}
\end{figure}


\subsection{One-dimensional necklaces}

It is instructive to consider ``necklace solutions'' in one dimension. While less of physical interest,
they allow for simpler analysis and simulations than their two-dimensional counterparts.

Let $\psi(z,x)$ be the solution of the one-dimensional NLS on~$D = [x_{\rm L}, x_{\rm R}]$
\begin{subequations}
\label{eq:NLS_bounded-1D}
\begin{equation}
i\psi_z(z,x)+\psi_{xx}+|\psi|^{2 \sigma}\psi=0,
\qquad  x_{\rm L}<x< x_{\rm R}  , \qquad z>0,
\end{equation}
with the Dirichlet condition
\begin{equation}
\psi(z,x_{\rm L})=\psi(z,x_{\rm R})=0, \qquad z\geq0.
\end{equation}
\end{subequations}
Eq.~(\ref{eq:NLS_bounded-1D}) admits solitary wave solutions~$\psi_{\rm sw} =e^{i\mu z}R_\mu(x)$, where~$R_\mu$ is a solution of
\begin{equation}
 \label{eq:dDim_R_ODE2-1D}
 R_{xx}(x)-\mu R+|R|^{2 \sigma}R=0, \quad  x_{\rm L}<x< x_{\rm R},
\qquad R(x_{\rm L})=R(x_{\rm R})=0.
\end{equation}


Let~$D$ be the interval $[-1,1]$. Then equation~\eqref{eq:dDim_R_ODE2-1D}  reads
\begin{equation}
  \label{eq:Rlambdaline}
R_{xx} (x)-\mu R+|R|^{2 \sigma}R=0,\quad -1<x<1, \qquad R(\pm1)=0.
\end{equation}
The ground-state solitary waves $R^{(1)}_{\mu}(x)$
of~\eqref{eq:Rlambdaline} bifurcate from the ground-state linear mode
\begin{equation}
   \label{eq:Rlambdaline_b-lin}
Q^{(1)}(x) := \cos\left(\frac{\pi x}{2}\right), \quad \mu_{\rm lin} =
-\frac{\pi^{2}}{4}.
\end{equation}
 As~$\mu$
increases from~$\mu_{\rm lin}$, $R^{(1)}_{\mu}$ becomes
more localized as its power increases and the nonlinearity becomes more pronounced,
see Figure~\ref{fig:1dRx}. In particular, $ R^{(1)}_{\mu}\thicksim R_{\mu, \rm 1D}^{(1), \rm free}$ as $\mu
\to \infty$, where $R_{\mu, \rm 1D}^{(1), \rm free}$ is the positive solution of~\eqref{eq:R-1D-free}.

\begin{figure}[ht!]
\begin{center}
\scalebox{0.65}{\includegraphics{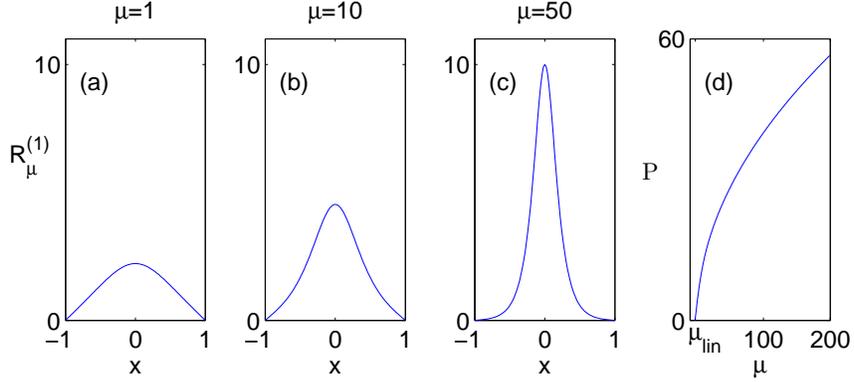}}
\caption{ A single pearl on an interval, i.e., the ground-state solution of~(\ref{eq:Rlambdaline}) with $\sigma=1$.
 (a)~$\mu=1$. (b)~$\mu=10$. (c)~$\mu=50$. (d)~The pearl power $P(R^{(1)}_{\mu})=\int_{-1}^{1}|R^{(1)}_{\mu}|^{2}dx$ as a function of~$\mu$. Here, $\mu_{\rm lin} \approx -2.5.$
     }
   \label{fig:1dRx}
\end{center}
\end{figure}


We can use $R^{(1)}_{\mu}$ to construct a one-dimensional necklace solitary wave with
$n$~pearls on the interval $[-n, n]$,
by letting adjacent pearls have opposite phases (signs).
For example, Fig.~\ref{fig:one-dimensional_necklace} shows a one-dimensional necklace solitary
wave with $3$ pearls.

\begin{figure}[h!]
\begin{center}
\scalebox{0.5}{\includegraphics{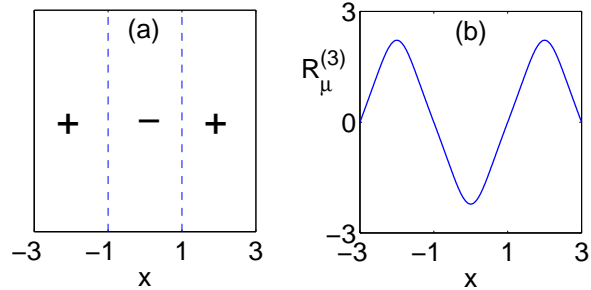}}
\caption{a)~Construction of a one-dimensional necklace solitary wave with $3$~pearls.
 The symbols $'\pm'$ correspond to $\pm R^{(1)}_{\mu}$. Note that $R^{(3)}_{\mu}$ is identically zero at $x = \pm 1$, and is antisymmetric with respect to these points.
(b)~The one-dimensional necklace solitary wave~$R_{\mu=1}^{(3)}$.
}
 \label{fig:one-dimensional_necklace}
\end{center}
\end{figure}

\section{Stability}
\label{sec:Linear stability}

The appropriate notion for stability of NLS solitary waves is that of orbital stability.
On a bounded domain, orbital stability refers to stability up to phase shifts:

\begin{Def}[orbital stability]
  Let $\psi_{\rm sw} = e^{ i \mu z} R_\mu$ be a solitary wave solution of~\eqref{eq:NLS_bounded}. 
  We say that~$\psi_{\rm sw}$ is orbitally stable, if for any $\epsilon>0$, there exists
$\delta>0$, such that if
$
   \|{\psi}_0  - R_\mu \|_{H_0^1(D)} < \delta,
$
and if $\psi$ is the solution of~\eqref{eq:NLS_bounded}
with the initial condition~$\psi_0$,
then
$$
\inf_{\theta(z)\in {\Bbb R}}\|{\psi}(z, {\bf x}) - e^{i \theta(z)} \psi_{\rm sw}(z, {\bf x}) \|_{H_0^1(D)} < \epsilon,
  \qquad  0 \le z < \infty.
$$
\end{Def}

\subsection{Stability of single-pearl (ground-state)  solitary waves}
\label{sec:Linear stability_single_pearl}

We recall that in free space, the ground-state solitary waves of
the critical NLS~\eqref{eq:NLS_free_space} are unstable.
A reflecting boundary, however, has a stabilizing effect. Indeed, since 
the VK condition $\frac{d}{d\mu}P(R_{\mu})>0$ 
holds for the ground state solitary waves on~$B_1$,
these solitary waves are orbitally stable for $\mu_{\rm lin}<\mu<\infty$ ($0<P<\Pcr$). See~\cite{bounded-01,Fukuizumi-12} for more details.

\subsubsection{Convex domains}

Our numerical simulations suggest that the VK condition holds for the ground state  solutions of~\eqref{eq:dDim_R_ODE2}
on a square, on a sector of a circle, and on a sector of an annulus
(see Fig.~\ref{fig:2dRxy}(d), Fig.~\ref{fig:quarter_circle}(d), and Fig.~\ref{fig:quarter_circle_hole}(d), respectively). In addition, solving the NLS~\eqref{eq:NLS_bounded} with an initial condition
which is a perturbed single pearl suggests
that  single-pearl solitary waves are stable, see e.g., Fig.~\ref{fig:rectangle_dx0_025dy0_025Zmax10Xmax1Ymax1mu1only0and10}.
Therefore, based on this numerical evidence, we formulate
 \begin{conj}
 All single-pearl (positive, ground state) solitary waves $\psi_{\rm sw}^{(1)} = e^{i \mu z} R^{(1)}_{\mu}(x,y)$
of the {\rm NLS}~\eqref{eq:NLS_bounded} on a convex bounded domain
 are orbitally stable.
\end{conj}
  \begin{figure}[ht!]
\begin{center}
\scalebox{0.6}{\includegraphics{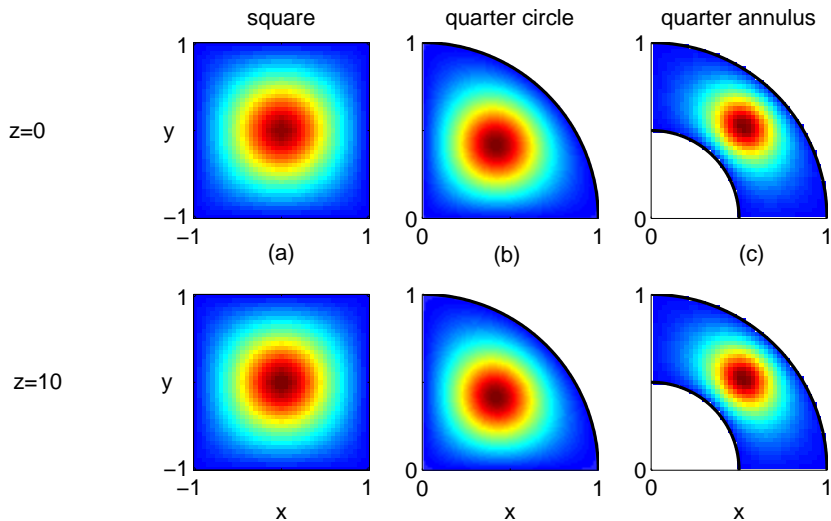}}
\caption{ Contour plots of $|\psi|$ at $z=0$ (top) and $z=10$ (bottom). Here $\psi$ is the solution
 of the NLS~\eqref{eq:NLS_bounded} on a domain~$D$, with the perturbed single-pearl initial condition $\psi_{0}=1.05 R_{\mu=1}^{(1)}(x,y)$.
(a)~$D$ is the square $[-1, 1]^2$. (b)~$D$ is the quarter circle $\{0\leq r \leq 1\ , 0\leq \theta \leq \frac{\pi}{2}\}$.
(c)~$D$ is the quarter annulus $\{\frac{1}{2}\leq r \leq 1\ , 0\leq \theta \leq \frac{\pi}{2}\}$. }
\label{fig:rectangle_dx0_025dy0_025Zmax10Xmax1Ymax1mu1only0and10}
\end{center}
\end{figure}

\subsubsection{Annular domain}

In the case of an annular domain, there are two types of positive solitary waves:
Radial ring-type solutions $\psi_{\rm sw}^{\rm ring} = e^{ i \mu z} R^{\rm ring}_\mu(r)$,
which exist for $ \mu_{\rm lin} <\mu< \infty$, and
nonradial single-peak solutions $\psi_{\rm sw}^{(1)} = e^{ i \mu z} R^{(1)}_\mu(r,\theta)$,
which exist for $ \mu_{\rm c} <\mu< \infty$, see
Section~\ref{sec:Annular-radial-solitary-waves}.
The ground states solutions (i.e., the minimizers of the Hamiltonian, see Conjecture~\ref{conj:unique_min}) are
radial for $ \mu_{\rm lin} <\mu< \mu_{\rm c}$ and non-radial for $ \mu_{\rm c} <\mu< \infty$.

To test numerically for stability, we solve the NLS with an initial condition which is 
the solitary wave profile, multiplied by~1.05 in the first quadrant. Thus,     
$\psi_0 = R^{\rm ring}_\mu(r) H(\theta)$ or $\psi_0 = R^{(1)}_\mu(r) H(\theta)$,
where  
\begin{equation}
\label{eq:H}
H(\theta) = 
\left\{ 
\begin{array}{cl}
1.05, \quad &\mbox{for~~} 0 <\theta < \pi/2, \\
1 , \quad   &\mbox{for~~}   \pi/2 <\theta <2\pi.
\end{array}
\right.
\end{equation}
Numerical simulations (Figure~\ref{fig:R_full_annular_dr0_0125dth0_16_Zmax_10_Rmax_1_Rmin0_5_mu_-38_5and-37_9_power_5_percent__quarter_pertubation} top and bottom raws) show that the ground state solitary waves $\psi_{\rm sw}^{(1)}$
are stable both below and above the symmetry-breaking point~$\mu_{\rm cr}$
(i.e.,
for $ \mu_{\rm lin} <\mu < \infty $).
The ring-type solitary waves $\psi_{\rm sw}^{\rm ring}$, however,
are unstable for $ \mu_{\rm cr} <\mu< \infty $, see
Figure~\ref{fig:R_full_annular_dr0_0125dth0_16_Zmax_10_Rmax_1_Rmin0_5_mu_-38_5and-37_9_power_5_percent__quarter_pertubation} (middle raw).
Thus, as expected, the ground-state solutions are stable, but the excited ones are not.\footnote{A similar stability pattern, below and above the symmetry-breaking point, was observed for the inhomogeneous NLS~\cite{Kirr-08}.}'\footnote{Intuitively, the high-power
ring-type solitary waves inherit the azimuthal instability of the free-space ring-type solutions.
}

\begin{figure}[ht!]
\begin{center}
\scalebox{0.8}{\includegraphics{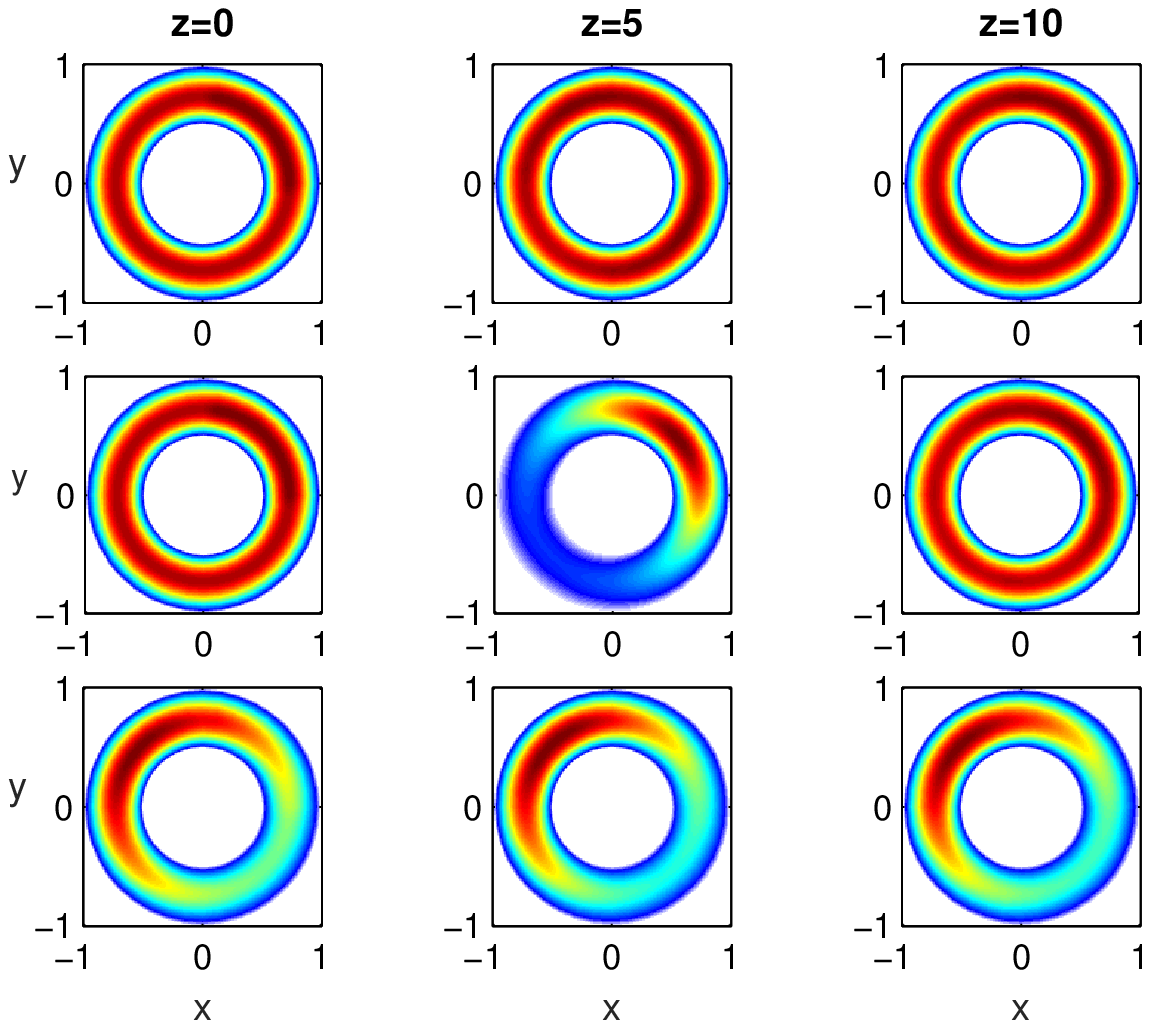}}
\caption{Solution of the NLS~\eqref{eq:NLS_bounded} on the annular domain $D = \{0.5 \le |{\bf x}| \le 1\}$. Top row:~$\mu=\mu_{\rm cr}-0.3$ and $\psi_0 = R^{\rm ring}_\mu(r) H(\theta)$, where  $H(\theta)$ is given by~\eqref{eq:H}.
 Middle row:~$\mu=\mu_{\rm cr}+0.3$
 and $\psi_0 = R^{\rm ring}_\mu(r) H(\theta)$.
 Bottom row:~$\mu=\mu_{\rm cr}+0.3$
 and $\psi_0 = R^{(1)}_\mu(r,\theta) H(\theta)$.
Here $\mu_{\rm cr} \thickapprox-38.2$.
}
\label{fig:R_full_annular_dr0_0125dth0_16_Zmax_10_Rmax_1_Rmin0_5_mu_-38_5and-37_9_power_5_percent__quarter_pertubation}
\end{center}
\end{figure}

\subsection{Stability of necklace solitary waves}
   \label{sec:Necklace-stability}

 Since each pearl satisfies the VK condition (Section~\ref{sec:Linear stability_single_pearl}),
so does the necklace, see~\eqref{eq:P=nP}. This does not imply that the necklace is stable, however, since the VK condition implies stability only for ground states. Therefore, to check for necklace stability, we have to ``go back'' to the original eigenvalue problem from which the VK condition was derived~\cite{Vakhitov-73}.

Let
$$
\psi = e^{i \mu t}(R({\bf x}) + \epsilon h(z,{\bf x})), \qquad h(z,{\bf x}) = e^{ \Omega z} (h(z,{\bf x}) +i v({\bf x})).
$$
Then the linearized equation for~$h$ reads
\begin{equation}
    \label{eq:eigenvalues2}
\left({\begin{array}{cc}
0 & L_{+} \\
-L_{-} & 0
\end{array}} \right)
\left(\begin{array}{c} v \\
u \end{array} \right)= \Omega
\left(\begin{array}{c} v \\
u \end{array} \right),
 \end{equation}
where $\Omega \in \mathbb{C}$ is the eigenvalue, $\left(\begin{array}{c} v \\ u \end{array} \right)=\left(\begin{array}{c} v(x,y) \\
u(x,y) \end{array} \right)$ is the eigenfunction, and
$$
L_{+}:=\Delta -\mu +3|R_\mu|^{2} ,\quad  L_{-}:=\Delta-\mu +|R_\mu|^{2}.
$$
Note that if~$\Omega$ is an eigenvalue with eigenvector~$\left(\begin{array}{c} v \\ u \end{array} \right)$, then $-\Omega$ is
an eigenvalue with eigenvector
$\left(\begin{array}{c} -v \\ u \end{array} \right)$,
 and $\Omega^{*}$ is an eigenvalue with eigenvector
$\left(\begin{array}{c} v^{*} \\ u^{*} \end{array} \right)$.
 Therefore, the genuinely-complex eigenvalues of~(\ref{eq:eigenvalues2}) always
come in quadruples.
Consequently, without loss of generality,
we focus from now on eigenvalues with nonnegative real and imaginary parts.

Necklace solitary waves are stable under perturbations that preserve their symmetries (or, more precisely, their anti-symmetries with respect to the interfaces between adjacent pearls):
\begin{lem}
   \label{lem:necklace-stability}
Let $\psi_{\rm sw}^{(n)}$ be a necklace solitary wave with $n$~pearls.
If the single-pearl solitary wave is stable, then~$\psi_{\rm sw}^{(n)}$ is stable under perturbations that preserve
the symmetries of~$\psi_{\rm sw}^{(n)}$.
\end{lem}
\begin{proof}
As in free space, see Lemma~\ref{lem:necklace} and~\eqref{eq:psi===0}, the solution remains zero on the interfaces between pearls. Therefore, there is no interaction between pearls. Hence, if the single pearl is stable, then so is the necklace.
\end{proof}

\Remark As $\mu \to \infty$, the single pearl has the asymptotic scaling
$$
R^{(1)}_{\mu}\thicksim R^{(1), \rm free}_{\mu, \rm 2D}(r) = \mu^{\frac{1}{2}}R^{(1), \rm free}_{\rm 2D}(\mu^{\frac{1}{2}}r).
$$ Therefore, the necklace $R^{(n)}_{\mu}(x,y)$
has the same asymptotic scaling.
Hence, by~\eqref{eq:eigenvalues2},
so does the eigenvector~$\left(\begin{array}{c} v \\ u \end{array} \right)$,
while the eigenvalue~$\Omega$ scales asymptotically as~$\mu$.

 \subsection{Rectangular necklaces}
   \label{sec:stability-rectangle}

  When $R_\mu=R_\mu^{(1)}$ is a single pearl on the square $D=[-1, 1]^2$, there are no eigenvalues
of~\eqref{eq:eigenvalues2} with a positive real part, see Fig.~\ref{fig:Re_Omega_as_meu_2_d_4_rects_sigma_1_Nx_41_Ny_41}(a).
This is in agreement with the observed stability of $\psi_{\rm sw}^{(1)}=e^{i \mu z}R^{(1)}_{\mu}$
in Fig.~\ref{fig:rectangle_dx0_025dy0_025Zmax10Xmax1Ymax1mu1only0and10}(a).

  Next, we investigate the stability of $\psi_{\rm sw}^{(2 \times 1)}=R_\mu^{(2\times 1)}(x,y)e^{i\mu z}$,
where $R_\mu^{(2\times 1)}$ is a rectangular necklace with 2~pearls.
 Fig.~\ref{fig:Re_Omega_as_meu_2_d_4_rects_sigma_1_Nx_41_Ny_41}(b) shows that there are no eigenvalues
of~\eqref{eq:eigenvalues2}
with $\mbox{Re}(\Omega)>0$ for $\mu_{\rm lin}\leqslant \mu < \mu_{\rm cr}$,
and that there is a simple eigenvalue~$\Omega_1(\mu)$ with $\mbox{Re}(\Omega_1)>0$ for $ \mu_{\rm cr}<\mu<\infty$,
where  $\mu_{\rm lin}=-\frac{\pi^{2}}{2}\thickapprox -4.9$, see~\eqref{eq:Rlambdarectangular_b-lin},
and~$\mu_{\rm cr}\thickapprox-4$.
 Hence, $\psi_{\rm sw}^{(2\times 1)}$ is linearly stable for $\mu_{\rm lin}\leqslant \mu < \mu_{\rm cr}$ and unstable for $ \mu_{\rm cr}<\mu<\infty$.
As noted, $R^{(2\times 1)}$~satisfies the VK condition for all~$\mu$, since it is constructed from $2$ identical ground states,
each of which satisfies the VK condition.
This does not lead to a contradiction, however, since the VK condition applies to ground states,
which is not the case for
multi-pearl necklaces.

 The change from stability to instability occurs as the necklace power exceeds the threshold power
\begin{equation}
  \label{eq:Pth_square_n=2}
P_{\rm th}^{\rm necklace}(n=2\times 1) := P(R_{\mu_{\rm cr}}^{(2\times 1)})
\approx 0.275P_{\rm cr},
\end{equation}
see Figure~\ref{fig:Re_Omega_as_P_2_d_4_rects_sigma_1_Nx_41_Ny_41}(b),
 i.e., when the power of each pearl is $\approx 0.14 \Pcr$.
For comparison, all single-pearl
solitary waves with power below~$P_{\rm cr}$ are stable. From this perspective,
2-pearl necklace solitary waves are considerably less stable than single-pearl ones.
\begin{figure}[ht!]
\begin{center}
\scalebox{0.6}{\includegraphics{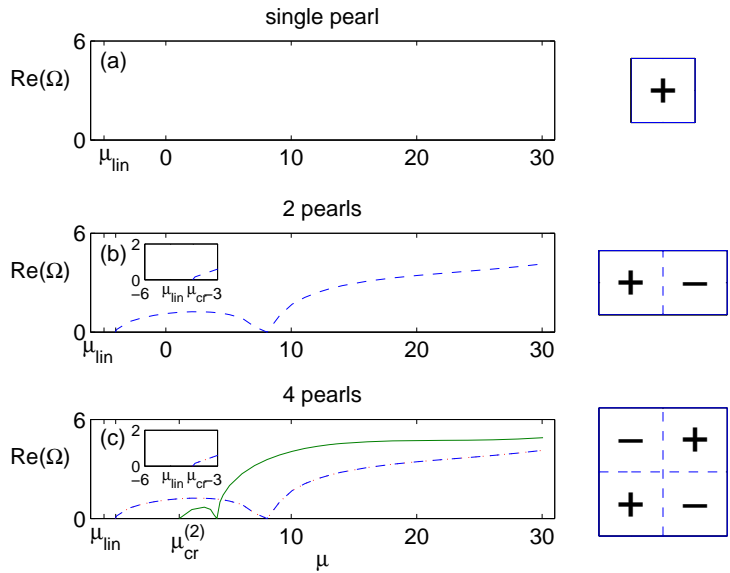}}
\caption{  $\mbox{Re}(\Omega)$ as a function of~$\mu$ for all eigenvalues of~(\ref{eq:eigenvalues2}) with a positive real part. Here~$D$ is the rectangular domain shown on the right.
 (a)~$R=R^{(1)}_{\mu}$.
(b)~$R=R^{(2 \times 1)}_{\mu}$. The dashed line corresponds to~$\Omega_1$.
(c)~$R=R^{(2 \times 2)}_{\mu}$.
The dashed blue and dotted red lines correspond to~$\Omega_1$, and the solid green line to~$\Omega_2$.
 }
\label{fig:Re_Omega_as_meu_2_d_4_rects_sigma_1_Nx_41_Ny_41}
\end{center}
\end{figure}

\begin{figure}[ht!]
\begin{center}
\scalebox{0.6}{\includegraphics{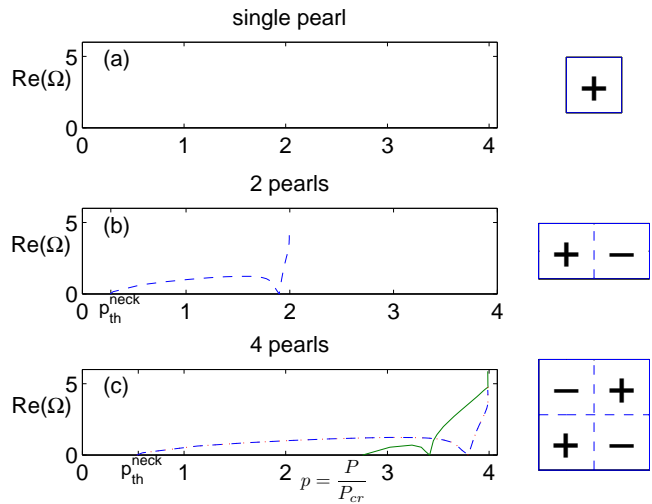}}
\caption{Same as Fig~\ref{fig:Re_Omega_as_meu_2_d_4_rects_sigma_1_Nx_41_Ny_41}, where the $x$-coordinate is the fractional necklace power $p = P(R_\mu^{(n)})/\Pcr$.
Here
$p_{\rm th}^{\rm neck} := P_{\rm th}^{\rm necklace}(n)/\Pcr$ is $\approx 0.275$ in~(b) and  $\approx 0.55$ in~(c).
 }
\label{fig:Re_Omega_as_P_2_d_4_rects_sigma_1_Nx_41_Ny_41}
\end{center}
\end{figure}

Following~\cite{Weinstein-06,Sivan-08}, in order to understand the {\em instability dynamics} of~$\psi_{\rm sw}^{(2\times 1)}$, in Fig~\ref{fig:eigenvectors_Re_Omega_as_meu_2rects_2d_Nx41Ny41_Xmax_1Ymax1mu_min20_mu_max_20_dmu_3_k_eigs_400_4o}
we plot the eigenfunction $f^{(2\times 1)}_1:=\left(\begin{array}{c} v \\ u \end{array} \right)$ that corresponds
to the unstable eigenvalue~$\Omega_1$ of~\eqref{eq:eigenvalues2}
with~$R=R_\mu^{(2 \times 1)}$. We plot both the imaginary and real parts (and not just the absolute value)
since we want to check whether~$f^{(2\times 1)}_1$ is symmetric or anti-symmetric.
The unstable eigenfunction satisfies $f^{(2\times 1)}_1(x,y)=f^{(2\times 1)}_1(-x,y)$, in contrast to~$R_\mu^{(2\times 1)}$ that satisfies $R_\mu^{(2\times 1)}(x,y)=-R_\mu^{(2\times 1)}(-x,y)$.
In other words, $f^{(2\times 1)}_1$ is symmetric with respect to the interface $x=0$ between the pearls, whereas
$R_\mu^{(2\times 1)}$ is antisymmetric with respect to $x=0$.
 This shows that the instability is related to the breaking of the anti-symmetry between the two pearls.
{\em Since the addition of a symmetric perturbation to an anti-symmetric profile lowers one of the peaks while increasing the other, the instability evolves as power flows from one pearl to the other.}

\begin{figure}[ht!]
\begin{center}
\scalebox{1}{\includegraphics{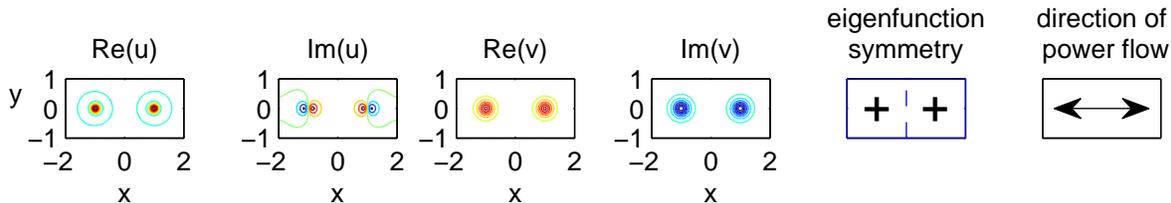}}
\caption{Real and imaginary parts of the eigenfunction $f^{(2\times 1)}_1 = {v \choose u}$
 of~(\ref{eq:eigenvalues2})
that corresponds to the unstable eigenvalue
$\Omega_1(\mu=20)=3.4+4.4i$.  Here $D  = [-2, 2] \times [-1, 1]$ and $R=R_{\mu=20}^{(2 \times 1)}$.
The rightmost panel illustrates the directions in which the power flows as the instability evolves.
 }
\label{fig:eigenvectors_Re_Omega_as_meu_2rects_2d_Nx41Ny41_Xmax_1Ymax1mu_min20_mu_max_20_dmu_3_k_eigs_400_4o}
\end{center}
\end{figure}

Next, we investigate the stability of $\psi_{\rm sw}^{(2 \times 2)}=R_\mu^{(2 \times 2)}(x,y)e^{i\mu z}$, where $R^{(2 \times 2)}$ is
a square necklace with $2 \times 2$ pearls.
  Fig.~\ref{fig:Re_Omega_as_meu_2_d_4_rects_sigma_1_Nx_41_Ny_41}(c) shows that $\psi_{\rm sw}^{(2 \times 2)}$ is stable for $\mu_{\rm lin}\leqslant \mu < \mu_{\rm cr}$ and unstable for $\mu_{\rm cr}<\mu<\infty$, where~$\mu_{\rm lin}$
and~$\mu_{\rm cr}$ are as in the case of~$R_\mu^{(2 \times 1)}$.
Furthermore, $\psi_{\rm sw}^{(2 \times 2)}$ has an unstable eigenvalue~$\Omega_1(\mu)$ of multiplicity~2
for $\mu_{\rm cr}
< \mu< \infty$, whose value is the same as for~$\psi_{\rm sw}^{(2 \times 1)}$,
and an additional simple eigenvalue~$\Omega_2(\mu)$ for $ \mu_{\rm cr}^{(2)}< \mu < \infty  $,
where $\mu_{\rm cr} <\mu_{\rm cr}^{(2)}\thickapprox 1$.

\begin{figure}[ht!]
\begin{center}
\scalebox{1}{\includegraphics{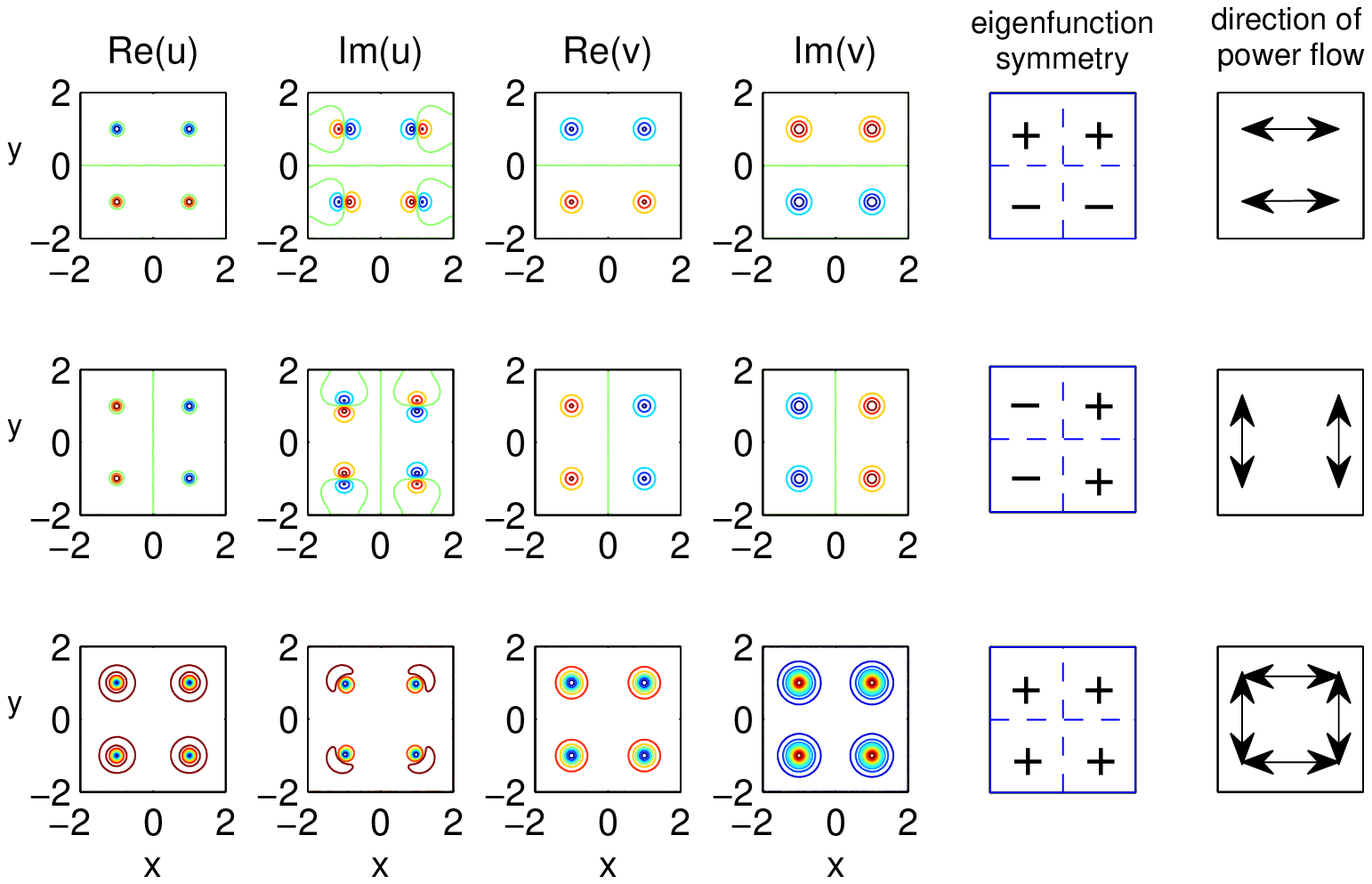}}
\caption{  Same as Fig.~\ref{fig:eigenvectors_Re_Omega_as_meu_2rects_2d_Nx41Ny41_Xmax_1Ymax1mu_min20_mu_max_20_dmu_3_k_eigs_400_4o} with $D  = [-2, 2]^2$ and 
$R=R_{\mu=20}^{(2 \times 2)}$. First row:~The eigenfunction $f^{(2 \times 2)}_{1,1}$ of $\Omega_1$.
Second row:~The eigenfunction $f^{(2 \times 2)}_{1,2}$ of $\Omega_1$.
Third row:~The eigenfunction $f^{(2 \times 2)}_{2}$ of $\Omega_2$.
}
\label{fig:eigenvectors_Re_Omega_as_meu_2d_Nx41Ny41_Xmax_1Ymax1mu_min20_mu_max_20_dmu_3_k_eigs_400_4o}
\end{center}
\end{figure}

In Fig.~\ref{fig:eigenvectors_Re_Omega_as_meu_2d_Nx41Ny41_Xmax_1Ymax1mu_min20_mu_max_20_dmu_3_k_eigs_400_4o} we plot the eigenfunctions that correspond to the unstable eigenvalues of~$R_\mu^{(2 \times 2)}$ for~$\mu=20$,  which is in the regime  $ \mu_{\rm cr}^{(2)}< \mu < \infty$.
One eigenfunction that corresponds to $\Omega_1(\mu=20)=3.4+4.4i$ satisfies $f^{(2 \times 2)}_{1,1}(x,y)=f^{(2 \times 2)}_{1,1}(-x,y)$ and $f^{(2 \times 2)}_{1,1}(x,y)=-f^{(2 \times 2)}_{1,1}(x,-y)$. In contrast, $R^{(2 \times 2)}_{\mu}(x,y)=-R^{(2 \times 2)}_{\mu}(-x,y)$ and $R^{(2 \times 2)}_{\mu}(x,y)=-R^{(2 \times 2)}_{\mu}(x,-y)$. Thus, $f^{(2 \times 2)}_{1,1}$ breaks the anti-symmetry of~$R^{(2 \times 2)}_{\mu}$ in the $x$-direction,
but preserves it in the $y$-direction.
It follows that there is a decoupling between the top and bottom pair of pearls as the instability evolves.
In fact, as illustrated in Figure~\ref{fig:f_2_2_rectangle}(a),
$f^{(2 \times 2)}_{1,1}(x,y)$ is nothing but  $f^{(2 \times 1)}_1(x,y)$ for $-2\leq y\leq 0,$ and $-f^{(2 \times 1)}_1(-x,-y)$ for $0\leq y\leq 2$, where $f^{(2\times 1)}_1$ is the unstable eigenfunction of~\eqref{eq:eigenvalues2} with $R^{(2 \times 1)}_{\mu}$, which is plotted in
Figure~\ref{fig:eigenvectors_Re_Omega_as_meu_2rects_2d_Nx41Ny41_Xmax_1Ymax1mu_min20_mu_max_20_dmu_3_k_eigs_400_4o}.
This explain why the values of~$\mu_{\rm cr}$ and~$\Omega_1(\mu)$ for~$R^{(2 \times 1)}_{\mu}$
and for~$R^{(2 \times 2)}_{\mu}$ are identical,
and also implies that
\begin{equation}
  \label{eq:P(R_cr)}
P_{\rm th}^{\rm necklace}(n=2 \times 2):=P(R_{\mu_{\rm cr}}^{(2 \times 2)})= 2 \cdot  P_{\rm th}^{\rm necklace}(n=2 \times 1) \approx 0.55P_{\rm cr},
\end{equation}
see Figure~\ref{fig:Re_Omega_as_P_2_d_4_rects_sigma_1_Nx_41_Ny_41}(c),
where $P_{\rm th}^{\rm necklace}(n=2 \times 2)$ is the threshold power for instability of a $2 \times 2$ square necklace.

The second eigenfunction that corresponds to $\Omega_1(\mu=20)=3.4+4.4i$
is a $90^{\circ}$ rotation of the first eigenfunction, i.e.,
$f^{(2 \times 2)}_{1,2}(x,y)$ is $f^{(2 \times 1)}_1(-y,x)$ for $-2\leq x\leq 0$
  and $-f^{(2 \times 1)}_1(y,-x)$ for $0\leq x\leq 2$,
 see Figure~\ref{fig:f_2_2_rectangle}(b). Thus, $f^{(2 \times 2)}_{1,2}$ preserves the antisymmetry in~$x$,
but breaks it in~$y$. Hence, in the instability dynamics
 which is excited by this eigenfunction, there is a decoupling between the left and right pairs of pearls.

The unstable eigenfunction which corresponds to $\Omega_2(\mu=20)=4.7+4.6i$, which is denoted by~$f^{(2 \times 2)}_2$,
satisfies $f^{(2 \times 2)}_2(x,y)=f^{(2 \times 2)}_2(-x,y)=f^{(2 \times 2)}_2(x,-y)=f^{(2 \times 2)}_2(-x,-y)$.
Hence, it breaks the antisymmetry between all pearls. Therefore,
the power flows between all $4$~pearls.

\begin{figure}[ht!]
\begin{center}
\scalebox{0.6}{\includegraphics{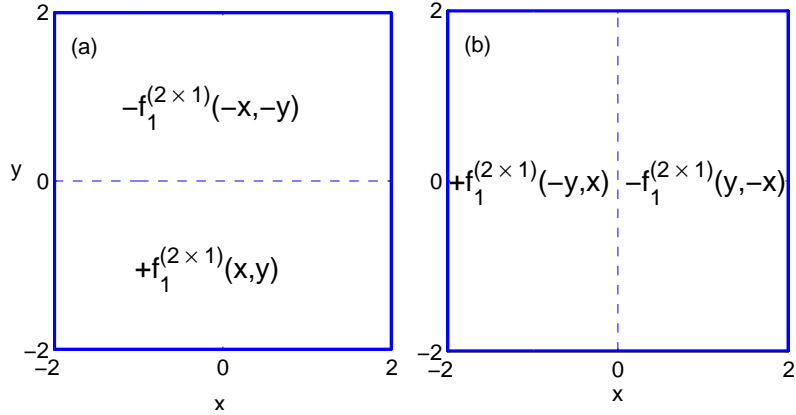}}
\caption{ (a)~Construction of~$f^{(2 \times 2)}_{1,1}$ using~$f^{(2\times 1)}_1$.
(b)~Construction of~$f^{(2 \times 2)}_{1,2}$ using~$f^{(2\times 1)}_1$. See details in text.
}
\label{fig:f_2_2_rectangle}
\end{center}
\end{figure}

\subsubsection{Instability dynamics}

 In free space, unstable solitary waves either collapse or scatter. On bounded domains, however, scattering is impossible.
In addition, since $P_{\rm th}^{\rm necklace}(n=2 \times 2) \approx 0.55P_{\rm cr}$, see~\eqref{eq:P(R_cr)},
then for~$\mu$ slightly above~$\mu_{\rm cr}$ we have that $P(R_{\mu}^{(2 \times 2)}) < P_{\rm cr}$, and so
$\psi_{\rm sw}^{(2 \times 2)}$
cannot collapse under small perturbations. Indeed,
{\em the change from stability to instability of~$\psi_{\rm sw}^{(2 \times 2)}$
is associated with the loss
of the necklace structure,\footnote{i.e., the loss of antisymmetry between adjacent pearls.} due to power flow between pearls, and not with collapse or scattering}. To see this,
in Fig.~\ref{fig:four_rectangles_dx0_025dy0_025_Zmax10_Xmax1_Ymax1_lamda-17_pertubated_5_percent_one_pearl_k_1_4o}
we solve the NLS numerically with a perturbed $R_\mu^{(2 \times 2)}$ initial condition.
In order to excite the unstable modes, we multiply the top-right pearl
by~$1.05$ and leave the other three pearls unchanged.\footnote{Unlike Figure~\ref{fig:rectangle_dx0_025dy0_025Zmax10Xmax1Ymax1mu1only0and10}, we do not test for stability by multiplying~$\psi_{0}$ by~1.05, since in that case
$\psi_0$ remains antisymmetric with respect to the interfaces between the pearls, and hence so does~$\psi$. Consequently, $\psi \equiv 0$ on the interfaces between pearls, which thus act as reflecting boundaries that prevent the pearls from interacting with each other.
 Since each pearl is stable, see Figure~\ref{fig:rectangle_dx0_025dy0_025Zmax10Xmax1Ymax1mu1only0and10}(a), the necklace remains stable
 under this antisymmetry-preserving perturbation (see Lemma~\ref{lem:necklace-stability}).}
 Fig.~\ref{fig:four_rectangles_dx0_025dy0_025_Zmax10_Xmax1_Ymax1_lamda-17_pertubated_5_percent_one_pearl_k_1_4o} shows that
  $\psi_{\rm sw}^{(2 \times 2)}$ preserves its necklace shape till $z=10$ for $\mu = \mu_{\rm cr}-\frac{1}{4}$,
but disintegrates before $z=3$ for $\mu = \mu_{\rm cr}+\frac{1}{4}$.
This confirms that there is a qualitative change in the necklace stability at~$\mu_{\rm cr}$.

\begin{figure}[ht!]
\begin{center}
\scalebox{0.6}{\includegraphics{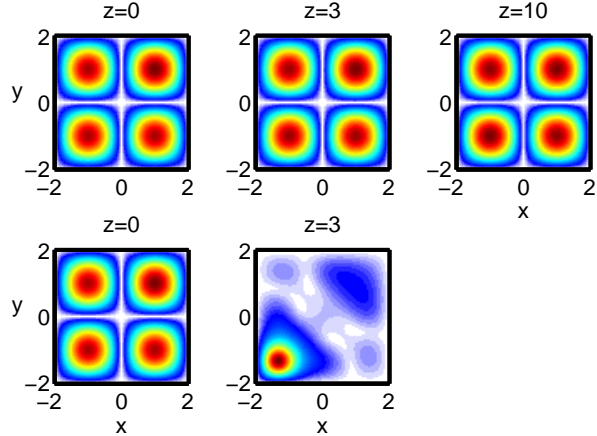}}
\caption{Solution of the NLS~(\ref{eq:NLS_bounded}) on the domain $D  = [-2, 2]^2$ at different distances. Here $\psi_{0}=R^{(2 \times 2)}_{\mu}(x,y)H(\theta)$, 
where  $H(\theta)$ is given by~\eqref{eq:H}. Top row:~$\mu=\mu_{\rm cr}-\frac{1}{4}$. Bottom row:~$\mu=\mu_{\rm cr}+\frac{1}{4}$. Here $\mu_{\rm cr}\thickapprox-4$.  }
\label{fig:four_rectangles_dx0_025dy0_025_Zmax10_Xmax1_Ymax1_lamda-17_pertubated_5_percent_one_pearl_k_1_4o}
\end{center}
\end{figure}


\subsection{Circular necklaces}
  \label{sec:stability-circle}

When $R_\mu=R_\mu^{(1)}$ is a single pearl on the quarter circle $\{0\leq r \leq 1\ , 0\leq \theta \leq \frac{\pi}{2}\}$, there are no eigenvalues of~\eqref{eq:eigenvalues2} with a positive real part, see Fig.~\ref{fig:Re_Omega_as_meu_2_d_4_quarters_circle_sigma_1_Nr_40_Nth_40}(a). This is in agreement with the observed stability of $\psi_{\rm sw}^{(1)}=e^{i \mu z}R^{(1)}_{\mu}$
in Fig.~\ref{fig:rectangle_dx0_025dy0_025Zmax10Xmax1Ymax1mu1only0and10}(b).

  Next, we investigate the stability of $\psi_{\rm sw}^{(2)}=R_\mu^{(2)}(x,y)e^{i\mu z}$,
where $R_\mu^{(2)}$ is a  2-pearl necklace on the semi-circle.
 Fig.~\ref{fig:Re_Omega_as_meu_2_d_4_quarters_circle_sigma_1_Nr_40_Nth_40}(b) shows that there are no eigenvalues
of~\eqref{eq:eigenvalues2}
with $\mbox{Re}(\Omega)>0$ for $\mu_{\rm lin}\leqslant \mu < \mu_{\rm cr}$,
and there is a single eigenvalue~$\Omega_1$ with $\mbox{Re}(\Omega_1)>0$ for $ \mu_{\rm cr}<\mu<\infty$,
where  $\mu_{\rm lin}=-k_{2}^{2}\thickapprox-26.4$, see~\eqref{eq:2Dim_R_ODE_general-lin}, and
 $\mu_{\rm cr}\thickapprox-24.3$.
 Hence, $\psi_{\rm sw}^{(2)}$ is linearly stable for $\mu_{\rm lin}\leqslant \mu < \mu_{\rm cr}$ and unstable for $ \mu_{\rm cr}<\mu<\infty$.\footnote{
As noted, $R^{(2)}$~satisfies the VK condition for all~$\mu$, since it is constructed from two identical ground states,
each of which satisfies the VK condition.
This does not lead to a contradiction, however, since the VK condition applies to positive solutions, which is not the case for
multi-pearl necklaces.}

 The change from stability to instability occurs as the necklace power exceeds the threshold power

\begin{equation}
  \label{eq:Pth_circular_n=2}
P_{\rm th}^{\rm necklace}(n=2) = P(R_{\mu_{\rm cr}}^{(2)})
\approx 0.12P_{\rm cr},
\end{equation}
see Figure~\ref{fig:Re_Omega_as_P_2_d_4_quarters_circle_sigma_1_Nr_40_Nth_40}(b),
 i.e., when the power of each pearl is $\approx  0.06 \Pcr$.
Thus, {\em the threshold power for instability of a 2-pearl necklace on the semi-circle is less than half of that on a rectangle},
see~\eqref{eq:Pth_square_n=2}. We have no explanation for this surprising observation.

\begin{figure}[ht!]
\begin{center}
\scalebox{0.6}{\includegraphics{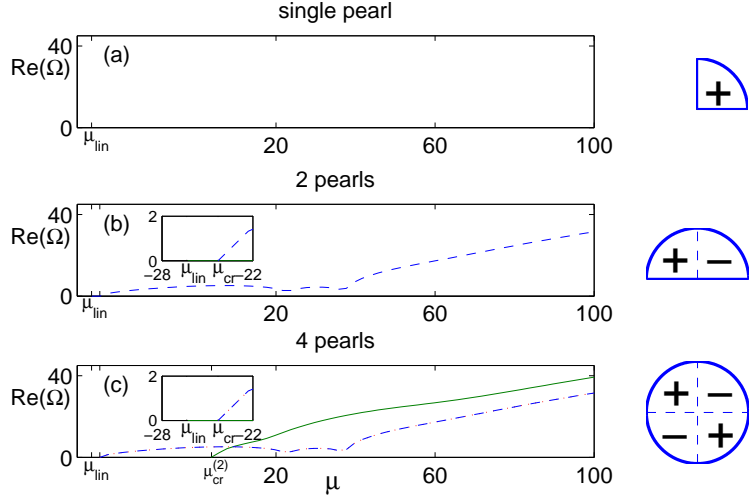}}
\caption{Same as Figure~\ref{fig:Re_Omega_as_meu_2_d_4_rects_sigma_1_Nx_41_Ny_41}, where $D$ is the sector of a circle shown on the right.
 (a)~$R=R^{(1)}_{\mu},$ (b)~$R=R^{(2)}_{\mu},$ (c)~$R=R^{(4)}_{\mu}.$
}
\label{fig:Re_Omega_as_meu_2_d_4_quarters_circle_sigma_1_Nr_40_Nth_40}
\end{center}
\end{figure}

\begin{figure}[ht!]
\begin{center}
\scalebox{0.6}{\includegraphics{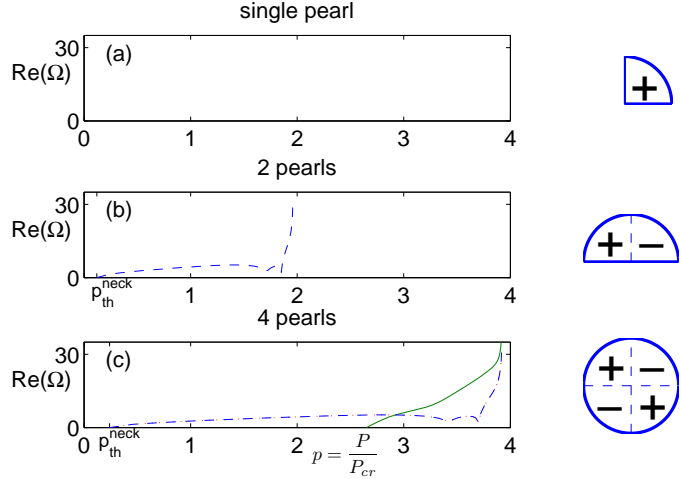}}
\caption{Same as Fig~\ref{fig:Re_Omega_as_meu_2_d_4_quarters_circle_sigma_1_Nr_40_Nth_40}, where the $x$-coordinate is the fractional necklace power $p = P(R_\mu^{(n)})/\Pcr$.
Here
$p_{\rm th}^{\rm neck} := P_{\rm th}^{\rm necklace}(n)/\Pcr$ is $\approx 0.12$ in~(b) and  $\approx 0.24$ in~(c).
 }
\label{fig:Re_Omega_as_P_2_d_4_quarters_circle_sigma_1_Nr_40_Nth_40}
\end{center}
\end{figure}

In Figure~\ref{fig:eigenvectors_two_quarters_Nr41Nth40mu12Rmax1k_eigs_400}
we plot the eigenfunction $f^{(2)}_1:=\left(\begin{array}{c} v \\ u \end{array} \right)$ that corresponds
to the unstable eigenvalue~$\Omega_1$
associated with~$R_\mu^{(2)}$. The unstable eigenfunction satisfies $f^{(2)}_1(x,y)=f^{(2)}_1(-x,y)$, 
in contrast to~$R_\mu^{(2)}$, which satisfies $R_\mu^{(2)}(x,y)=-R_\mu^{(2)}(-x,y)$. This shows that the instability
of~$\psi_{\rm sw}^{(2)}$ is related to the breaking of the anti-symmetry between the two pearls.
Thus, the instability evolves as power flows from one pearl to the other.

\begin{figure}[ht!]
\begin{center}
\scalebox{0.8}{\includegraphics{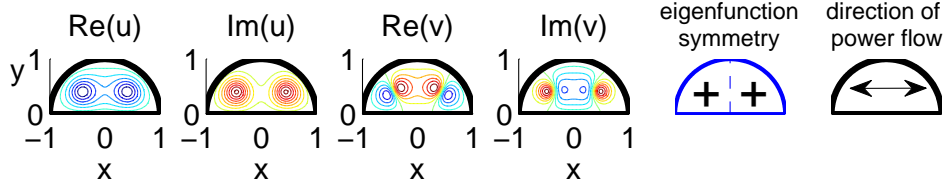}}
\caption{Same as Fig.~\ref{fig:eigenvectors_Re_Omega_as_meu_2rects_2d_Nx41Ny41_Xmax_1Ymax1mu_min20_mu_max_20_dmu_3_k_eigs_400_4o} with  $D$ being  the upper half circle, $R=R^{(2)}_{\mu=12}$,   and $\Omega_1(\mu=12) =5+14i$.
 }
\label{fig:eigenvectors_two_quarters_Nr41Nth40mu12Rmax1k_eigs_400}
\end{center}
\end{figure}

Next, we investigate the stability of $\psi_{\rm sw}^{(4)}=R_\mu^{(4)}(x,y)e^{i\mu z}$, where $R^{(4)}$ is
a circular necklace with $4$ pearls.
  Fig.~\ref{fig:Re_Omega_as_meu_2_d_4_quarters_circle_sigma_1_Nr_40_Nth_40}(c) shows that $\psi_{\rm sw}^{(4)}$ is linearly stable for $\mu_{\rm lin}\leqslant \mu < \mu_{\rm cr}$ and unstable for $\mu_{\rm cr}<\mu<\infty$, where $\mu_{\rm lin}$ and $\mu_{\rm cr}$ are as in the case of $R_\mu^{(2)}$.
Furthermore, $\psi_{\rm sw}^{(4)}$ has an unstable eigenvalue~$\Omega_1$ of multiplicity~2
for $\mu_{\rm cr}
< \mu< \infty$, whose value is the same as for~$\psi_{\rm sw}^{(2)}$,  and an additional simple eigenvalue~$\Omega_2$ for $ \mu_{\rm cr}^{(2)}< \mu < \infty  $, where $\mu_{\rm cr}<\mu_{\rm cr}^{(2)}\thickapprox 4$.

\begin{figure}[ht!]
\begin{center}
\scalebox{0.8}{\includegraphics{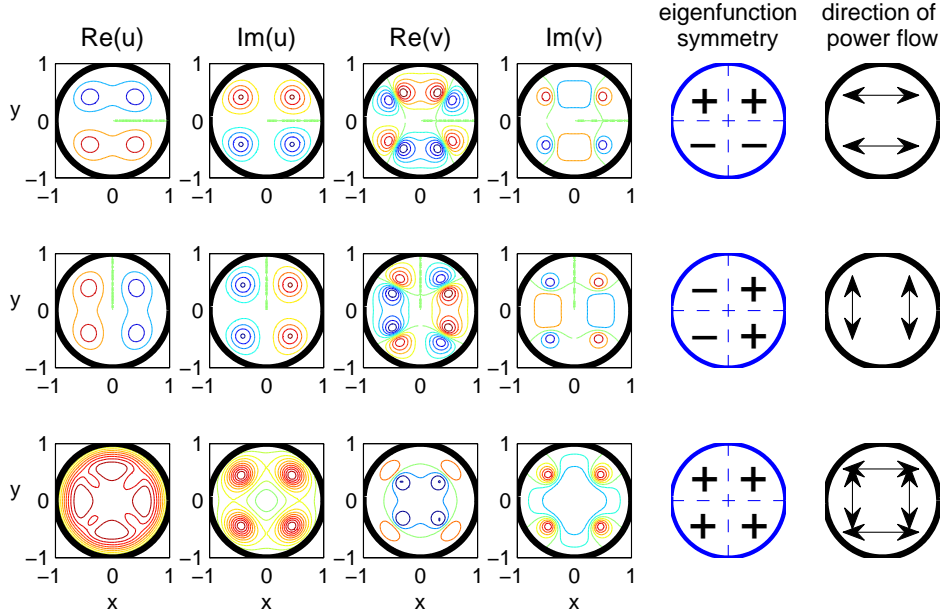}}
\caption{Same as Fig.~\ref{fig:eigenvectors_Re_Omega_as_meu_2rects_2d_Nx41Ny41_Xmax_1Ymax1mu_min20_mu_max_20_dmu_3_k_eigs_400_4o}, with~$D$ being the unit circle and $R=R^{(4)}_{\mu=12}$.
First row:~The eigenfunction $f^{(4)}_{1,1}$ of $\Omega_1$.
Second row:~The eigenfunction $f^{(4)}_{1,2}$ of $\Omega_1$.
Third row:~The eigenfunction $f^{(4)}_{2}$ of $\Omega_2$.
} 
\label{fig:eigenvectors_four_quarters_Nr41Nth40mu12Rmax1k_eigs_400_v2}
\end{center}
\end{figure}

In Fig.~\ref{fig:eigenvectors_four_quarters_Nr41Nth40mu12Rmax1k_eigs_400_v2} we plot the eigenfunctions that correspond to the unstable eigenvalues for $\mu=12$,  which is in the regime  $ \mu_{\rm cr}^{(2)}< \mu < \infty$.
As in the rectangular case,
one eigenfunction that corresponds to~$\Omega_1(\mu=12)=5+14i$ satisfies $f^{(4)}_{1,1}(x,y)=f^{(4)}_{1,1}(-x,y)$ and $f^{(4)}_{1,1}(x,y)=-f^{(4)}_{1,1}(x,-y)$. In contrast, $R^{(4)}_{\mu}(x,y)=-R^{(4)}_{\mu}(-x,y)$ and $R^{(4)}_{\mu}(x,y)=-R^{(4)}_{\mu}(x,-y)$. Thus, $f^{(4)}_{1,1}$ breaks the anti-symmetry of~$R^{(4)}_{\mu}$ in the $x$-direction,
but preserves it in the $y$-direction.
Hence, there is decoupling between the top and bottom pair of pearls as the instability evolves.
In fact, as illustrated in Figure~\ref{fig:f_2_2_circle}(a),
$f^{(4)}_{1,1}(x,y)$ is nothing but  $f^{(2)}_1(x,y)$ for $0\leq y\leq 1,$ and $-f^{(2)}_1(-x,y)$ for $-1\leq y\leq 0$, where $f^{(2)}_1$ is the unstable eigenfunction of~\eqref{eq:eigenvalues2} with $R^{(2)}$, which is plotted in
Figure~\ref{fig:eigenvectors_two_quarters_Nr41Nth40mu12Rmax1k_eigs_400}.
This explain why the values of~$\mu_{\rm cr}$ and~$\Omega_1(\mu)$ for~$R^{(2)}$ and for~$R^{(4)}_{\mu}$ are identical,
and also  implies that the threshold power for instability of a 4-pearl circular necklace satisfies
\begin{equation}
  \label{eq:P(R_cr)2}
P_{\rm th}^{\rm necklace}(n=4)= 2   P_{\rm th}^{\rm necklace}(n=2) \approx 0.24P_{\rm cr},
\end{equation}
see Figure~\ref{fig:Re_Omega_as_P_2_d_4_quarters_circle_sigma_1_Nr_40_Nth_40}(c).

The second eigenfunction that corresponds to~$\Omega_1(\mu=12)=5+14i$
is a $90^{\circ}$ rotation of the first eigenfunction, i.e.,
$f^{(4)}_{1,2}(x,y)$ is $-f^{(2)}_1(y,-x)$ for $-1\leq x\leq 0,$  and $f^{(2)}_1(-y,x)$ for $0\leq x\leq 1$,
 see Figure~\ref{fig:f_2_2_circle}(b).
 For this eigenfunction, there is decoupling between the instability dynamics of the left and right pairs of pearls.
The unstable eigenfunction which corresponds to~$\Omega_2(\mu=12)=6.7+19i$ is denoted by~$f^{(4)}_2$.
It satisfies $f^{(4)}_2(x,y)=f^{(4)}_2(-x,y)=f^{(4)}_2(x,-y)=f^{(4)}_2(-x,-y)$, i.e., it breaks the antisymmetry between all pearls. Therefore,
the power flows between all $4$~pearls.

\begin{figure}[ht!]
\begin{center}
\scalebox{0.6}{\includegraphics{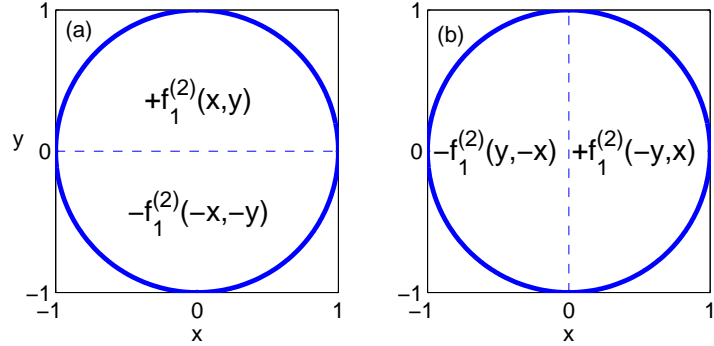}}
\caption{ (a)~Construction of $f^{(4)}_{1,1}$ using~$f^{(2)}_1$.
(b)~Construction of $f^{(4)}_{1,2}$ using~$f^{(2)}_1$. See details in the text.
}
\label{fig:f_2_2_circle}
\end{center}
\end{figure}

\subsubsection{Instability dynamics}

As in the rectangular case, the change from stability to instability of~$\psi_{\rm sw}^{(4)}$
is associated with the loss of the necklace structure due to power flow between pearls, and not with collapse or scattering. To see this,
we solve the NLS numerically with a perturbed $R_\mu^{(4)}$ initial condition.
In order to excite the unstable modes, we multiply the top-left pearl
by~$1.05$ and leave the other three pearls unchanged.
 Fig.~\ref{fig:R_four_quarters_dr0_025dth0_04_Zmax_10_Rmax_1_mu_-25_3_4o_pertubated5percent} shows that
$\psi_{\rm sw}^{(4)}$ preserves its necklace shape till $z=10$ for $\mu = \mu_{\rm cr}-1$, but that
it disintegrates before $z=3$ for $\mu = \mu_{\rm cr}-1$.
This confirms that there is a qualitative change in the necklace stability at~$\mu_{\rm cr}$.

\begin{figure}[ht!]
\begin{center}
\scalebox{0.6}{\includegraphics{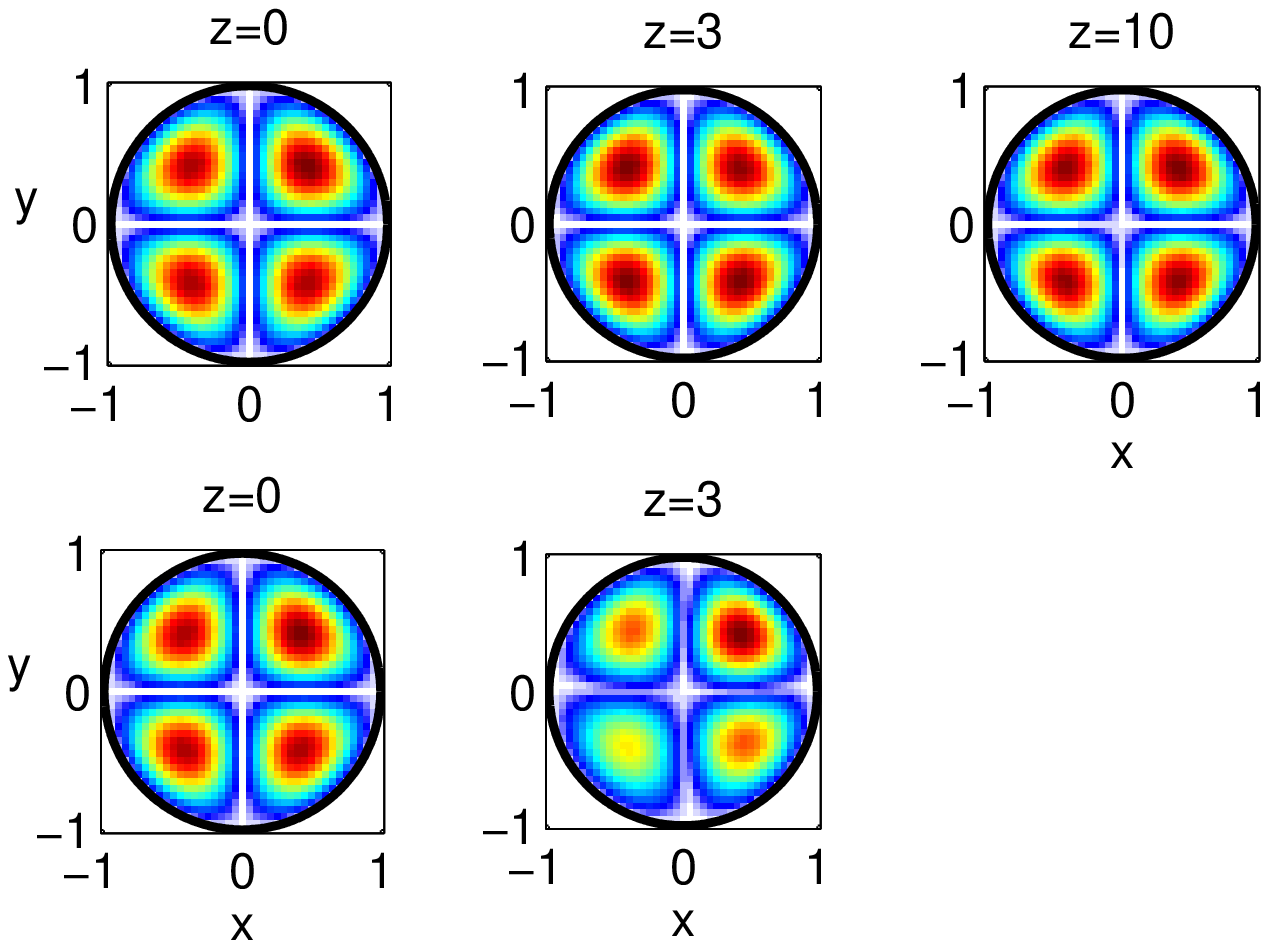}}
\caption{Same as Figure~\ref{fig:four_rectangles_dx0_025dy0_025_Zmax10_Xmax1_Ymax1_lamda-17_pertubated_5_percent_one_pearl_k_1_4o} on a circular domain. Top row:~$\mu=\mu_{\rm cr}-1$. Bottom row:~$\mu=\mu_{\rm cr}+1$. Here $\mu_{\rm cr}\thickapprox-24.3$. }
\label{fig:R_four_quarters_dr0_025dth0_04_Zmax_10_Rmax_1_mu_-25_3_4o_pertubated5percent}
\end{center}
\end{figure}

%
%
%

\subsection{Annular necklaces}

When $R_\mu=R_\mu^{(1)}$ is a single pearl on the quarter-annulus $\{0.5 \leq r \leq 1\ , 0\leq \theta \leq \frac{\pi}{2}\}$, there are no eigenvalues of~\eqref{eq:eigenvalues2} with a positive real part,
see Fig.~\ref{fig:Re_Omega_as_meu_2_d_4_quarters_circle_hole_sigma_1_Nr_41_Nth_40_new}(a). This is in agreement with the observed stability of $\psi_{\rm sw}^{(1)}=e^{i \mu z}R^{(1)}_{\mu}$
in Fig.~\ref{fig:rectangle_dx0_025dy0_025Zmax10Xmax1Ymax1mu1only0and10}(c).

  Next, we investigate the stability of $\psi_{\rm sw}^{(2)}=e^{i\mu z} R_\mu^{(2)}(x,y)$,
where $R_\mu^{(2)}$ is a  2-pearl necklace on the half-annulus
$\{0.5 \leq r \leq 1\ , 0\leq \theta \leq {\pi}\}$.
 Fig.~\ref{fig:Re_Omega_as_meu_2_d_4_quarters_circle_hole_sigma_1_Nr_41_Nth_40_new}(b) shows that there are no eigenvalues
of~\eqref{eq:eigenvalues2}
with $\mbox{Re}(\Omega)>0$ for $\mu_{\rm lin}\leqslant \mu < \mu_{\rm cr}$
and for $\mu_2\leqslant \mu < \mu_4$,
and there is a single eigenvalue~$\Omega_1$ with $\mbox{Re}(\Omega_1)>0$ for $ \mu_{\rm cr}<\mu<\mu_2$
and for $ \mu_4<\mu<\infty$,
where $\mu_{\rm lin}=-k_2^{2}\simeq-46.4$,
see~\eqref{eq:k_n/2}, $\mu_{\rm cr} \thickapprox-43.6$, $\mu_2 \approx 23$, and $\mu_4 \approx 117$.
 Hence, $\psi_{\rm sw}^{(2)}$ is linearly stable for $\mu_{\rm lin}\leqslant \mu < \mu_{\rm cr}$
and $\mu_2 < \mu < \mu_4$,
 and unstable for $ \mu_{\rm cr}<\mu<\mu_2$
and $ \mu_4<\mu<\infty$.

 The initial change from stability to instability occurs as the necklace power exceeds
$$
P_{\rm th}^{\rm necklace}(R_{\mu}^{(2)}) = P(R_{\mu_{\rm cr}}^{(2)})
\approx 0.12P_{\rm cr},
$$
see Figure~\ref{fig:Re_Omega_as_P_2_d_4_quarters_circle_hole_sigma_1_Nr_41_Nth_40_new}(b),
 i.e., when the power of each pearl is $\approx  0.06 \Pcr$.
Thus, {\em the threshold power for instability of a two-pearl necklace on the half-annulus is
the same as that on the half-circle}, see~\eqref{eq:Pth_circular_n=2}.
This is surprising, since the necklace instability is associated with power transfer between the pearls, and so the hole could be expected to have a stabilizing effect.



The behavior of~$\mbox{Re}(\Omega_1)$ for $\mu_{\rm lin}<\mu < \mu_4$ is
different from that for a two-pearl necklace on a rectangle and on a semi-circle, see Figures~\ref{fig:Re_Omega_as_meu_2_d_4_rects_sigma_1_Nx_41_Ny_41}(b) and~\ref{fig:Re_Omega_as_meu_2_d_4_quarters_circle_sigma_1_Nr_40_Nth_40}(b), respectively. It is, however,
similar to that for a two-pearl one-dimensional necklace, see
Figure~\ref{fig:Re_Omega_as_meu_1_d_2_lines_sigma_1_Nx_101}(b). Intuitively, this is because the Dirichlet boundary conditions at $r=r_{\min}$ and $r=r_{\max}$ ``clamp'' the necklace profile in the radial direction.
Consequently, the dependence of~$R_\mu^{(1)}$  on~$\mu$ is predominantly one-dimensional.
Once  $\mu > \mu_4$, however, the pearls become so narrow that they ``do not feel'' the radial walls.
Therefore, they approach the free-space two-dimensional ground state
$R^{(1), \rm free}_{\mu, \rm 2D}$.
Hence, $\mbox{Re}(\Omega_1)$ increases linearly in~$\mu$ (see Section~\ref{sec:Necklace-stability}).

We now discuss the stability of $\psi_{\rm sw}^{(4)}=R_\mu^{(4)}(r,\theta)e^{i\mu z}$,
where $R_\mu^{(4)}$ is an annular necklace with $4$~pearls,
see Sec.~\ref{sec:Annular necklace solitary waves}.
 Fig.~\ref{fig:Re_Omega_as_meu_2_d_4_quarters_circle_hole_sigma_1_Nr_41_Nth_40_new}(c)
 shows that $\psi_{\rm sw}^{(4)}$ is linearly stable for $\mu_{\rm lin}\leqslant \mu < \mu_{\rm cr}$
and $\mu_2\leqslant \mu < {\mu}_3$,
 and unstable for $ \mu_{\rm cr}<\mu<\mu_2$
and $ {\mu}_3<\mu<\infty$, where $\mu_{\rm lin}$, $\mu_{\rm cr}$, and $\mu_2$ are the same as for~$\psi_{\rm sw}^{(2)}$, and
$ {\mu}_3 \approx 73$.\footnote{The upper limit of the second stability regime is smaller than in the two-pearl case
(i.e., $\mu_3<\mu_4$),  since it is determined by the unstable eigenfunction~$f^{(4)}_2$ which corresponds to~$\Omega_2$. }

\begin{figure}[ht!]
\begin{center}
\scalebox{0.7}{\includegraphics{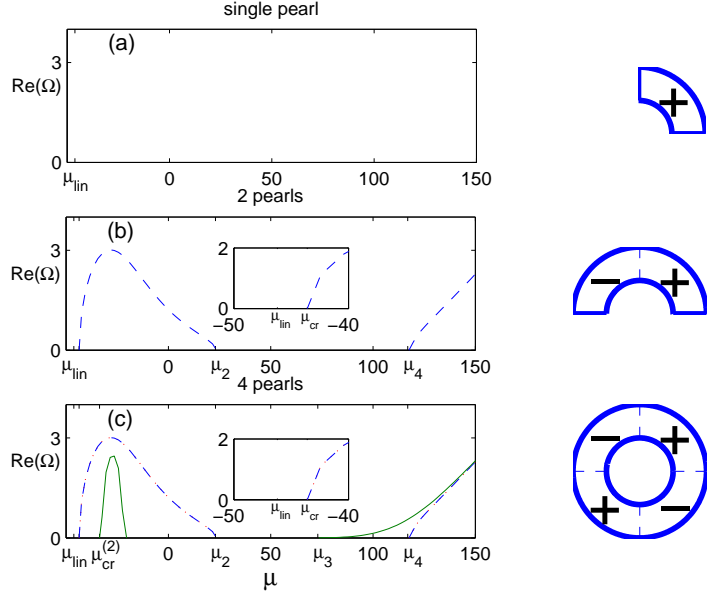}}
\caption{Same as Figure~\ref{fig:Re_Omega_as_meu_2_d_4_rects_sigma_1_Nx_41_Ny_41}, where $R=R^{(4)}_{\mu}$,
and $D$ is the sector of the annulus $0.5<r<1$ shown on the right.
}
\label{fig:Re_Omega_as_meu_2_d_4_quarters_circle_hole_sigma_1_Nr_41_Nth_40_new}
\end{center}
\end{figure}

\begin{figure}[ht!]
\begin{center}
\scalebox{0.7}{\includegraphics{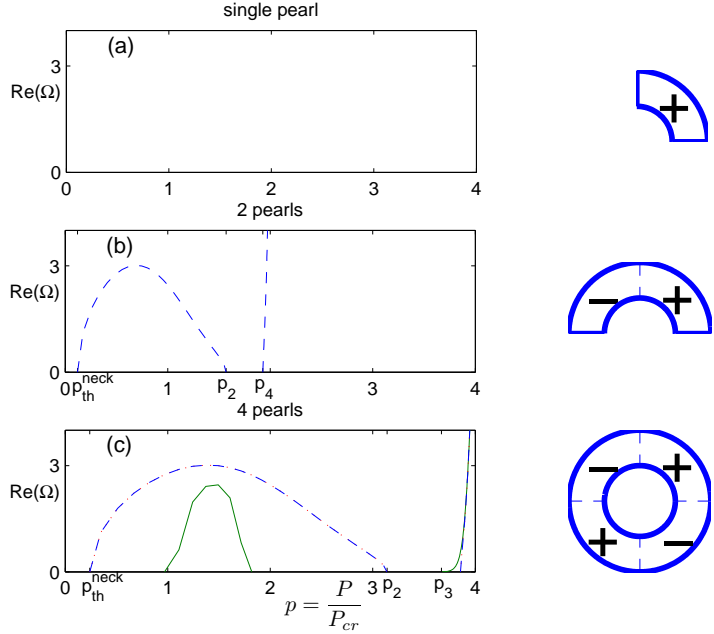}}
\caption{Same as Fig~\ref{fig:Re_Omega_as_meu_2_d_4_quarters_circle_hole_sigma_1_Nr_41_Nth_40_new}, where the $x$-coordinate is the fractional necklace power $p = P(R_\mu^{(n)})/\Pcr$. Here 
$p_{\rm th}^{\rm neck} = P_{\rm th}^{\rm necklace}(n)/\Pcr$ and $p_i:=P(R_{\mu_i}^{(n)})/\Pcr$.
 }
\label{fig:Re_Omega_as_P_2_d_4_quarters_circle_hole_sigma_1_Nr_41_Nth_40_new}
\end{center}
\end{figure}

In Fig.~\ref{fig:eigenvectors_four_quarters_hole_Nr41Nth40mu-28Rmax1Rmin0_5k_eigs_800} we plot the eigenfunctions that correspond to the unstable eigenvalues
for $\mu=-28$. 
One eigenfunction that corresponds to~$\Omega_1(\mu=-28)=3+7.6i$ satisfies $f^{(4)}_{1,1}(x,y)=f^{(4)}_{1,1}(-x,y)$ and $f^{(4)}_{1,1}(x,y)=-f^{(4)}_{1,1}(x,-y)$. In contrast, $R_\mu^{(4)}(x,y)=-R_\mu^{(4)}(-x,y)$ and $R_\mu^{(4)}(x,y)=-R_\mu^{(4)}(x,-y)$. Thus, $f^{(4)}_{1,1}$~breaks the anti-symmetry of~$R_\mu^{(4)}$ in the $x$-direction,
but preserves it in the $y$-direction.
Hence, there is decoupling between the top and bottom pair of pearls as the instability evolves.
As in the rectangular and circular cases,
$f^{(4)}_{1,1}(x,y)$ is nothing but
$f^{(2)}_1(x,y)$ for $-1\leq y\leq 0,$ and $-f^{(2)}_1(-x,-y)$ for $0\leq y\leq 1$
see Figure~\ref{fig:f_2_2_annular}(a),
where~$f^{(2)}_1$ is the unstable eigenfunction of~\eqref{eq:eigenvalues2}
for a two-pearl necklace on the half-annulus.
The second eigenfunction that corresponds to~$\Omega_1(\mu=-28)=3+7.6i$
is a $90^{\circ}$ rotation of the first eigenfunction, i.e.,
$f^{(4)}_{1,2}(x,y)$ is $f^{(2)}_1(-y,x)$ for $-1\leq x\leq 0,$  and $-f^{(2)}_1(y,-x)$ for $0\leq x\leq 1$,
 see Figure~\ref{fig:f_2_2_annular}(b).
 For this eigenfunction, there is decoupling between the instability dynamics of the left and right pair of pearls.
The unstable eigenfunction that corresponds to~$\Omega_2(\mu=-28)=2.5+16i$ is denoted by~$f^{(4)}_2$.
It satisfies $f^{(4)}_2(x,y)=f^{(4)}_2(-x,y)=f^{(4)}_2(x,-y)=f^{(4)}_2(-x,-y)$, i.e., it breaks the antisymmetry between all pearls. Therefore,
the power flows between all $4$~pearls.

The initial threshold power for the 4-pearl annular necklace instability is
\begin{equation}
P_{\rm th}^{\rm necklace}(n=4) = P(R_{\mu_{\rm cr}}^{(4)}) 
\approx 0.24P_{\rm cr},
\end{equation}
see Figure~\ref{fig:Re_Omega_as_P_2_d_4_quarters_circle_hole_sigma_1_Nr_41_Nth_40_new}(c),
which is the same as in the circular case, see~\eqref{eq:P(R_cr)2}. As noted,
this result is surprising, as one could expect the hole to stabilize the necklace.

\begin{figure}[ht!]
\begin{center}
\scalebox{0.8}{\includegraphics{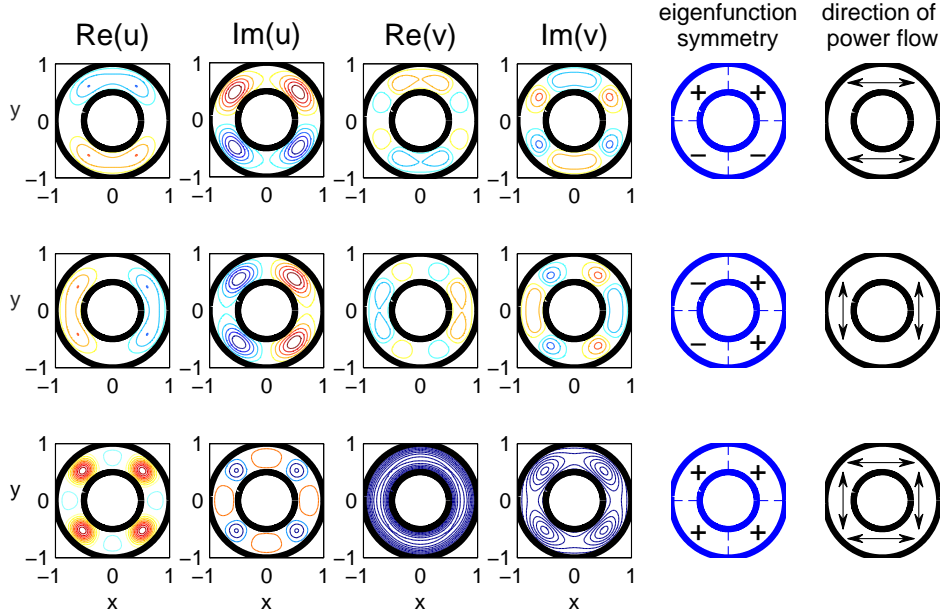}}
\caption{Same as Fig.~\ref{fig:eigenvectors_four_quarters_Nr41Nth40mu12Rmax1k_eigs_400_v2} for the annular domain $D = \{ 0.5 \le r \le 1\}$ and $R=R^{(4)}_{\mu=-28}$.
}
\label{fig:eigenvectors_four_quarters_hole_Nr41Nth40mu-28Rmax1Rmin0_5k_eigs_800}
\end{center}
\end{figure}

\begin{figure}[ht!]
\begin{center}
\scalebox{0.6}{\includegraphics{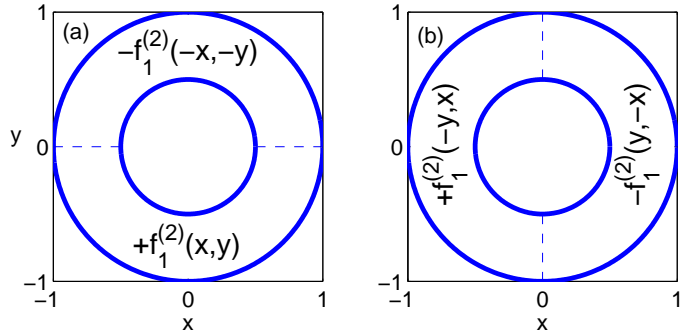}}
\caption{ (a)~Construction of $f^{(4)}_{1,1}$ using~$f^{(2)}_1$.
(b)~Construction of $f^{(4)}_{1,2}$ using~$f^{(2)}_1$. See details in text.
}
\label{fig:f_2_2_annular}
\end{center}
\end{figure}

As in the rectangular and circular cases, the initial change from stability to instability of~$\psi_{\rm sw}^{(4)}$
is associated with the loss of the necklace structure due to power flow between pearls, and not with collapse or scattering. To see this,
we solve the NLS numerically with a perturbed $R_\mu^{(4)}$ initial condition.
In order to excite the unstable modes, we multiply the upper-right pearl
by~$1.05$ and leave the other three pearls unchanged.
 Fig.~\ref{fig:R_four_quarters_dr0_0125dth0_04_Zmax_10_Rmax_1_Rmin0_5_mu_-44_6_pertubated_power_new} shows that
$\psi_{\rm sw}^{(4)}$ preserves its necklace shape till $z=10$ for $\mu = \mu_{\rm cr}-1$,
but  disintegrates before $z=3$ for $\mu = \mu_{\rm cr}+1$.
This confirms that there is a qualitative change in the necklace stability at~$\mu_{\rm cr}$.

\subsubsection{Second stability regime}

{\em Unlike the rectangular and circular cases,  there is a second regime in which the annular necklace becomes  linearly stable.}
This regime is $\mu_2 < \mu < \mu_3$,
see Figure~\ref{fig:Re_Omega_as_meu_2_d_4_quarters_circle_hole_sigma_1_Nr_41_Nth_40_new}(c),
which corresponds to
$$
p_2 \Pcr < P(R_\mu^{(4)}) < p_3\Pcr, \qquad
p_2 :=\frac{P(R_{\mu_2}^{(4)})}{\Pcr} \approx 3.1,  \quad
p_3 :=\frac{P(R_{\mu_3}^{(4)})}{\Pcr} \approx 3.7,
$$
see Figure~\ref{fig:Re_Omega_as_P_2_d_4_quarters_circle_hole_sigma_1_Nr_41_Nth_40_new}(c).
In this regime the hole does stabilize the necklace. Indeed,
Fig.~\ref{fig:R_four_quarters_dr0_0125dth0_04_Zmax_15_Rmax_1_Rmin0_5_mu_50_and 10_power_new_pertubation_2_percent_random_noise} shows that when perturbed by 5\% random noise,
$\psi_{\rm sw}^{(4)}$ becomes unstable and collapses around $z = 11.3$ for $\mu=10$,
but remains stable until $z=15$
for~$\mu=50$, where $P(R_{\mu=10}^{(4)}) \approx 2.9\Pcr$ and
 $P(R_{\mu=50}^{(4)}) \approx 3.5\Pcr$.

%

\begin{figure}[ht!]
\begin{center}
\scalebox{0.6}{\includegraphics{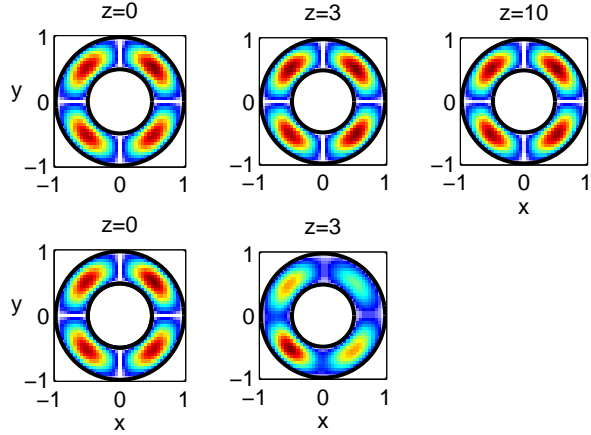}}
\caption{Same as Figure~\ref{fig:four_rectangles_dx0_025dy0_025_Zmax10_Xmax1_Ymax1_lamda-17_pertubated_5_percent_one_pearl_k_1_4o} for the annular domain $0.5 \le r \le 1$. Top row:~$\mu=\mu_{\rm cr}-1$. Bottom row:~$\mu=\mu_{\rm cr}+1$.
Here $\mu_{\rm cr} \thickapprox-43.6$.
}
\label{fig:R_four_quarters_dr0_0125dth0_04_Zmax_10_Rmax_1_Rmin0_5_mu_-44_6_pertubated_power_new}
\end{center}
\end{figure}

\begin{figure}[ht!]
\begin{center}
\scalebox{0.8}{\includegraphics{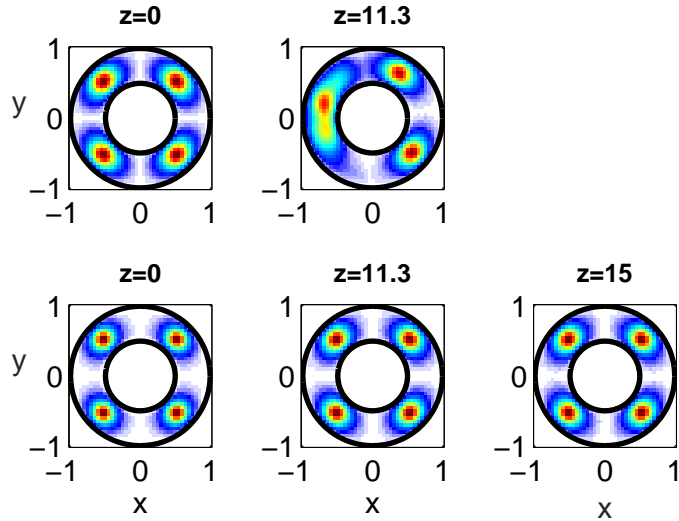}}
\caption{Same as Fig.~\ref{fig:R_four_quarters_dr0_0125dth0_04_Zmax_10_Rmax_1_Rmin0_5_mu_-44_6_pertubated_power_new} for $\psi_{0} = (1+0.05 \cdot \mbox{noise}(x,y))R^{(4)}_{\mu}$, where $\mbox{noise}(x,y)$
is uniformly distributed in~$[-1, 1]$. Top row:~$\mu=10$. Bottom row:~$\mu=50$.
 }
\label{fig:R_four_quarters_dr0_0125dth0_04_Zmax_15_Rmax_1_Rmin0_5_mu_50_and 10_power_new_pertubation_2_percent_random_noise}
\end{center}
\end{figure}

\subsection{One-dimensional necklaces}
\label{sec:One-dimensional-necklaces}

We now consider the stability of necklace solutions of the one-dimensional NLS~\eqref{eq:NLS_bounded-1D}.
We simultaneously consider the subcritical case $\sigma=1$ and the critical case  $\sigma = 2$, in order to
emphasize that, as far as necklace stability is concerned, the subcritical and critical cases are more similar
than different.

Fukuizumi et al.~\cite{Fukuizumi-12} rigorously proved that
when $0<\sigma \le 2$ and~$R^{(1)}_\mu$ is a single pearl on the interval $[-1,1]$,
the solitary waves $\psi_{\rm sw}^{(1)}(z,x) = e^{i \mu z} R_\mu^{(1)}(x)$
of~\eqref{eq:NLS_bounded-1D} are stable for all $\mu_{\rm lin}<\mu<\infty$.
Indeed, there are no eigenvalues of \eqref{eq:eigenvalues2} with a positive real part,
see Fig.~\ref{fig:Re_Omega_as_meu_1_d_2_lines_sigma_1_Nx_101}(a) and~\ref{fig:Re_Omega_as_meu_1_d_2_lines_sigma_2_Nx_101}(a).\footnote{In the one-dimensional case $L_{+}:=\frac{\partial^2}{\partial x^2} -\mu +(2 \sigma+1)|R|^{2}$ and
$L_{-}:=\frac{\partial^2}{\partial x^2}-\mu +|R|^{2}$.
}

 When $R=R_\mu^{(2)}$, a one-dimensional necklace with 2~pearls,
there are no eigenvalues of~\eqref{eq:eigenvalues2}
with $\mbox{Re}(\Omega)>0$ for $\mu_{\rm lin}\leqslant \mu < \mu_{\rm cr}$,
and there is a single eigenvalue~$\Omega_1$ with $\mbox{Re}(\Omega_1)>0$ for $ \mu_{\rm cr}<\mu<\infty$,
see Fig.~\ref{fig:Re_Omega_as_meu_1_d_2_lines_sigma_1_Nx_101}(b)
 and~\ref{fig:Re_Omega_as_meu_1_d_2_lines_sigma_2_Nx_101}(b),
where  $\mu_{\rm lin}=-\frac{\pi^{2}}{4}\thickapprox-2.5$, see~\eqref{eq:Rlambdaline_b-lin},
$\mu_{\rm cr}\thickapprox-1.56$ for $\sigma=1$, and
$\mu_{\rm cr}\thickapprox-2.16$  for $\sigma=2$.
 Hence, $\psi_{\rm sw}^{(2)}=R_\mu^{(2)}(x)e^{i\mu z}$
is linearly stable for $\mu_{\rm lin}\leqslant \mu < \mu_{\rm cr}$ and unstable for $ \mu_{\rm cr}<\mu<\infty$.

\begin{figure}[ht!]
\begin{center}
\scalebox{0.6}{\includegraphics{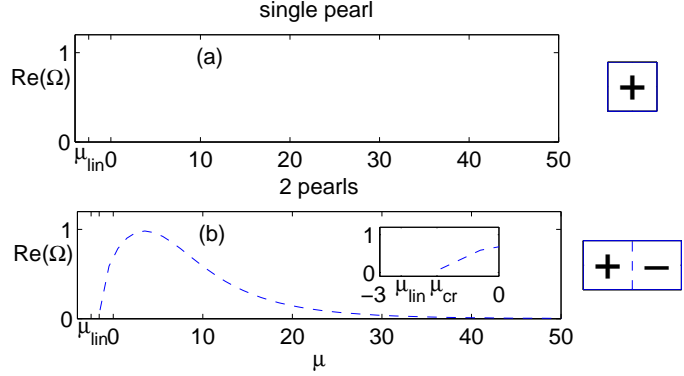}}
\caption{Same as Figure~\ref{fig:Re_Omega_as_meu_2_d_4_rects_sigma_1_Nx_41_Ny_41}
for the one-dimensional NLS with $\sigma=1$.
 (a)~$D = [-1, 1]$ and $R=R^{(1)}_{\mu}$.
(b)~$D = [-2, 2]$  and $R=R^{(2)}_{\mu}$.
}
\label{fig:Re_Omega_as_meu_1_d_2_lines_sigma_1_Nx_101}
\end{center}
\end{figure}

\begin{figure}[ht!]
\begin{center}
\scalebox{0.6}{\includegraphics{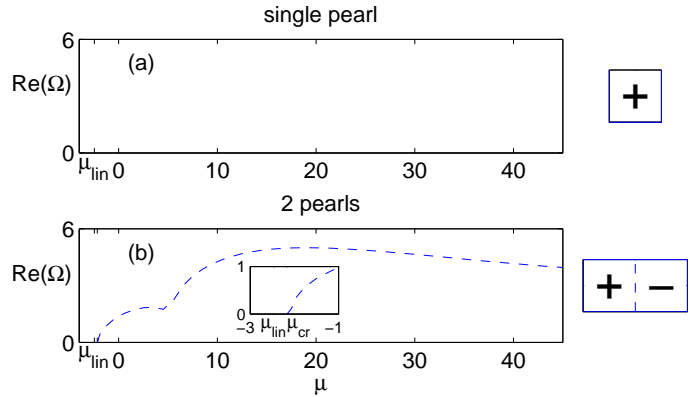}}
\caption{Same as Figure~\ref{fig:Re_Omega_as_meu_1_d_2_lines_sigma_1_Nx_101} for~$\sigma=2$.
}
\label{fig:Re_Omega_as_meu_1_d_2_lines_sigma_2_Nx_101}
\end{center}
\end{figure}

The eigenfunction $f^{(2)}_1:=\left(\begin{array}{c} v \\ u \end{array} \right)$ that corresponds
to the unstable eigenvalue~$\Omega_1$
associated with~$R_\mu^{(2)}$
satisfies $f^{(2)}_1(x)=f^{(2)}_1(-x)$,
see Figures~\ref{fig:eigenvectors_two_lines_Nx101mu10Xmax2_sigma1_full_eigs}
and~\ref{fig:eigenvectors_two_lines_Nx101mu10Xmax2_sigma2_full_eigs},
in contrast to~$R_\mu^{(2)}$, which satisfies $R_\mu^{(2)}(x)=-R_\mu^{(2)}(-x)$. Thus, the instability of~$\psi_{\rm sw}^{(2)}$ is related to the breaking of the anti-symmetry between the two pearls,
and it evolves as power flows from one pearl to the other. As in the two-dimensional case,
the necklace instability in unrelated to collapse. This is obvious in the
case of the one-dimensional cubic NLS, which is subcritical and thus does not admit blowup solutions.
In the one-dimensional quintic NLS, which is critical,
 the change from stability to instability occurs as the power of~$R_\mu^{(2)}$  exceeds
$P_{\rm th}^{\rm necklace}(R_\mu^{(2)}) = P(R_{\mu_{\rm cr}}^{(2)})
\approx 0.25P_{\rm cr}$,\footnote{i.e., when the power of each pearl is $\approx \frac18 \Pcr$.},\footnote{This result was already observed in~\cite{Fukuizumi-12}.}
which is well below the critical power for collapse.

\begin{figure}[ht!]
\begin{center}
\scalebox{0.7}{\includegraphics{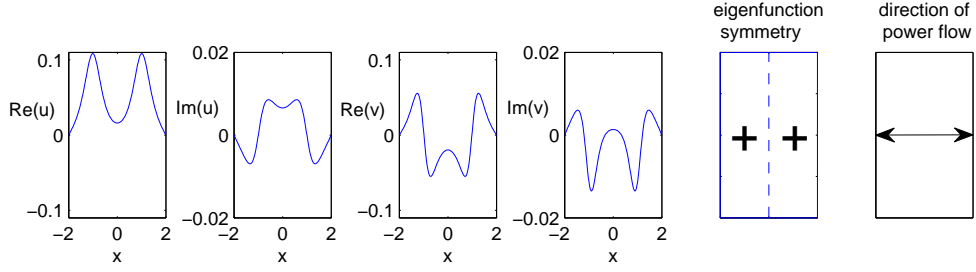}}
\caption{Same as Figure~\ref{fig:eigenvectors_Re_Omega_as_meu_2rects_2d_Nx41Ny41_Xmax_1Ymax1mu_min20_mu_max_20_dmu_3_k_eigs_400_4o} for the one-dimensional NLS with $\sigma=1$, $D = [-2, 2]$, $R=R_{\mu=10}^{(2)}$,  and
$\Omega_1(\mu=10)=0.59+2.2i$. }
\label{fig:eigenvectors_two_lines_Nx101mu10Xmax2_sigma1_full_eigs}
\end{center}
\end{figure}

\begin{figure}[ht!]
\begin{center}
\scalebox{0.7}{\includegraphics{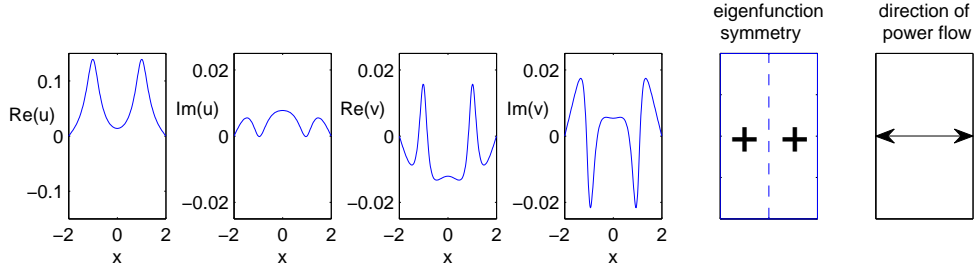}}
\caption{Same as Figure~\ref{fig:eigenvectors_two_lines_Nx101mu10Xmax2_sigma1_full_eigs} for~$\sigma=2$ and $\Omega_1(\mu=10)=4.3+5.4i$.  }
\label{fig:eigenvectors_two_lines_Nx101mu10Xmax2_sigma2_full_eigs}
\end{center}
\end{figure}

\section{Numerical methods}
   \label{sec:numerical}

The NLS~\eqref{eq:NLS_bounded} was integrated numerically using standard fourth-order finite differencing in~{\bf x},
with an implicit Crank-Nicolson method in~$z$ for the linear part, and a predictor-corrector method for the nonlinear part.
The eigenvalues and eigenvectors of~\eqref{eq:eigenvalues2} were calculated using Matlab's {\em eig} function.
The only numerical element which required a non-standard approach was the calculation of the solitary necklace profiles.

\subsection{Non-spectral renormalization method}
  \label{sec:Computation_of_the_solitary_waves}

A popular method for calculating multidimensional solitary waves is Petviashvili's spectral renormalization method~\cite{Petviashvili-76, Ablowitz-05}.
This method, however, can only compute positive solutions, which is not the case with necklace solutions.
 Therefore, we use  this method to compute a single pearl (which is positive), and then ``replicate'' it into a necklace, as explained in Section~\ref{sec:Solitary wave}.
 In addition, since Petviashvili's method is based on iterations in Fourier space, it cannot be applied on bounded domains. Therefore, following~\cite{Baruch_Fibich_Gavish:2009}, we employ a non-spectral renormalization method, as follows. We first rewrite~(\ref{eq:dDim_R_ODE2}) as
$$
   L(R)=|R|^{2}R, \qquad  (x,y) \in D, \qquad L:=-\Delta+\mu,
$$
subject to
$R=0$  on~$\partial D$. Then we consider the fixed-point iterations
\begin{equation}
     \label{eq:iterative_scheme}
R_{k+1}=L^{-1}\left[|R_{k}|^{2}R_{k}\right], \qquad k=0,1,\dots
\end{equation}
subject to
$R=0$  on~$\partial D$,
where $R_{k}=R_{k}(x,y)$ is the $k$th~iteration. To implement~\eqref{eq:iterative_scheme}, $L$~is discretized using
finite differences and inverted (once) using
the LU decomposition. We observe numerically that the
iterations (\ref{eq:iterative_scheme}) diverge to zero for a small
initial guess $R_{0}=R_{0}(x,y)$, and to infinity for a large initial guess. To
avoid this divergence, we renormalize the solution at each
iteration, so that it satisfies the integral relation\footnote{Other integral relations can also be used.},\footnote{Here, $\langle
f,g\rangle:=\int{f^{*} g \, dx dy}$ denotes the standard inner product.}
\begin{equation}\label{eq:integral_relation}
\mbox{SL}[R]=\mbox{SR}[R], \qquad \mbox{SL}[R]:=\langle R,R\rangle, \qquad \mbox{SR}[R]:=\langle R,L^{-1}\left(|R|^{2}R\right)\rangle,
\end{equation}
which is obtained from multiplication of
(\ref{eq:iterative_scheme}) by~$R$ and integration over~$(x,y)$. Thus, we define
$R_{k+\frac{1}{2}}:=c_{k}R_{k}$, where $c_{k}$ is chosen so
that $R_{k+\frac{1}{2}}$ satisfies~(\ref{eq:integral_relation}), i.e., $$c_{k}^{2} \mbox{SL}[R_{k}]=c_{k}^{4}\mbox{SR}[R_{k}].$$
 Consequently,
$c_{k}=\left({\mbox{SL}[R_{k}]}/{\mbox{SR}[R_{k}]}\right)^{\frac{1}{2}}$.
Therefore, the non-spectral renormalization method reads
\begin{equation}\label{eq:final_relation}
R_{k+1}=L^{-1}\left[|R_{k+\frac{1}{2}}|^{2}R_{k+\frac{1}{2}}\right]=\left(\frac{\langle
R_{k},R_{k}\rangle}{\langle
R_{k},L^{-1}\left(|R_{k}|^{2}R_{k}\right)\rangle}\right)
^{\frac{3}{2}}L^{-1}\left[|R_{k}|^{2}R_{k}\right],
\quad k=0,1,\dots
\end{equation}

To illustrate the convergence of this method,
consider the iterations~\eqref{eq:final_relation} when~$D$ is the square $[-1, 1]^2$, $\mu = 1$,
the initial guess  is $R_{0}=\sin\left(\pi (\frac{x-1}{2})\right)\sin\left(\pi (\frac{y-1}{2})\right)$, $L$ is discretized using fourth-order center-difference scheme, and $dx=dy=0.05$.
Fig.~\ref{fig:non_spectral_variant_of_Petviashvilis_method}(a) shows that after 40~iterations
$c_{k}$~converges to~1 with machine accuracy,  thus confirming the convergence of the iterations.
In order to check that the iterations converge to a solitary wave, we recall that the Pohozaev identities on a bounded domain read
\begin{eqnarray*}
\mu \|R_{\mu}\|_{2}^2 &=&
\frac{1}{2}
   \|R_{\mu}\|_{4}^{4}
  - \frac{1}{2}  \int_{\partial D}
 ({\bf x} \cdot {\bf n}) \, (\nabla R_{\mu} \cdot {\bf n})^2 \, d{\bf s},
\\
  \|\nabla R_{\mu}\|_{2}^2 &=&\frac{1}{2}
   \|R_{\mu}\|_{4}^{4}
  + \frac{1}{2}  \int_{\partial D}
 ({\bf x} \cdot {\bf n}) \, (\nabla R_{\mu} \cdot {\bf n})^2 \, d{\bf s},
\end{eqnarray*}
where~${\bf n}$ is the outward unit normal to~${\partial D}$.
To avoid computing boundary integrals, we check whether~$R_{k}$ satisfies the sum of these Pohozaev identities, i.e.,
whether $E_{\rm Pohozaev}(R_{k}) \to 0$, where
$$
E_{\rm Pohozaev}(R):= \mu\|R\|_{2}^{2}+\|\nabla R\|_{2}^{2}-\|R\|_{4}^{4}.
  $$
Fig.~\ref{fig:non_spectral_variant_of_Petviashvilis_method}(b) shows that
$E_{\rm Pohozaev}$ converges to~$O(10^{-3})$, and not to machine accuracy.
This is because the iterations converge to a solution of a discretized version of~(\ref{eq:dDim_R_ODE2}).
In other words, the $O(10^{-3})$ error of the Pohozaev identities
is determined by the discretization error and not by the convergence of the iterations.
Indeed, we verified that if we use smaller values of~$dx$ and~$dy$,
the limiting Pohozaev error reduces according to the order of the discretization scheme being used.

\begin{figure}[ht!]
\begin{center}
\scalebox{0.6}{\includegraphics{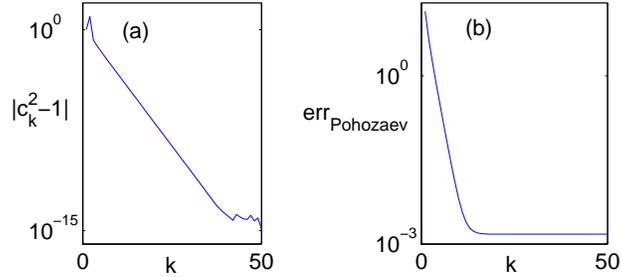}}
\caption{Convergence of the non-spectral renormalization method. (a)~$|c_{k}^{2}-1|$ as a function of the iteration number~$k$. (b)~${\rm err}_{\rm Pohozaev}(R_{k})$ as a function of the iteration number.
}
\label{fig:non_spectral_variant_of_Petviashvilis_method}
\end{center}
\end{figure}

{\bf Acknowledgments}  We thank A. Gaeta, B. Ilan, N. Fusco, and Y. Pinchover for useful suggestions. This research was partially supported by grant \#177/13 from the Israel Science Foundation (ISF),
and by the Kinetic Research Network (KI-Net) under the NSF Grant No. RNMS  \#1107444.

\bibliographystyle{plain}
\bibliography{WeakSol}
\end{document}